\documentclass[10pt]{IEEEtran}
\onecolumn
\pdfoutput=1

\usepackage{graphicx}
\usepackage{color}
\usepackage{amsmath}
\usepackage{algpseudocode}
\usepackage{algorithm}
\newtheorem{theorem}{Theorem}

\newtheorem{lemma}{Lemma}
\newtheorem{note}{Note}
\newtheorem{corollary}{Corollary}

\newtheorem{definition}{Definition}

\newcommand{\1}{{\bf 1}} 
\newcommand{\E}{\mathsf{E}} 

\newcommand{\card}[1]           {\left| #1\right|}
\newcommand{\Ascr}{{\cal A}}

\newcommand\argmax{\mathop{\mbox{{\rm argmax}}}\limits}

\newcommand{\bea}{\begin{eqnarray}}
\newcommand{\eea}{\end{eqnarray}}
\newcommand{\beas}{\begin{eqnarray*}}
\newcommand{\eeas}{\end{eqnarray*}}

\title{To Feed or Not to Feed Back}

\author{\IEEEauthorblockN{Himanshu Asnani\IEEEauthorrefmark{1}, Haim
Permuter\IEEEauthorrefmark{2} and 
Tsachy Weissman\IEEEauthorrefmark{3}} 
\thanks{\IEEEauthorblockA{\IEEEauthorrefmark{1}Stanford University, Email:
asnani@stanford.edu.}} 
\thanks{\IEEEauthorblockA{\IEEEauthorrefmark{2}Ben Gurion University, Email:
haimp@bgu.ac.il.}} 
\thanks{\IEEEauthorblockA{\IEEEauthorrefmark{3}Stanford University, Email:
tsachy@stanford.edu.}}}

\begin{document}
\maketitle

\begin{abstract}
We study the communication over Finite State Channels (FSCs), where the
encoder and the decoder can control the availability or the quality of the
noise-free feedback. Specifically, the instantaneous feedback is a function of
an action taken by the encoder, an action taken by the decoder, and the channel
output. Encoder and decoder actions take values in finite alphabets, and may be
subject to average cost constraints.
\par
We prove capacity
results
for such a setting by constructing a sequence of achievable rates, using a
simple scheme based on \textquoteleft \textit{code
tree}\textquoteright\ generation, that generates channel input symbols along 
with encoder and decoder actions. We prove that
the limit 
of this sequence exists. For a given block length $N$ and probability of
error, $\epsilon$, we give an upper bound on the maximum achievable rate. Our
upper and lower bounds coincide and hence yield the capacity for the case
where the probability of initial state
is
positive for all states. Further, for stationary
\textit{indecomposable}\ channels
without intersymbol 
interference (ISI), the capacity is given as the limit of normalized directed
information between the input and output sequence, maximized over an
appropriate set of causally conditioned distributions. As
an important special case, we consider the framework of
\textquoteleft\textit{to feed or not to feed back}\textquoteright\ where
either the encoder or the decoder takes binary actions, which determine whether
current channel output will be fed back to the encoder, with a constraint on the
fraction of channel outputs that are fed back. As another special case of our
framework, we characterize the capacity of \textquoteleft\textit{coding on
the backward link}\textquoteright\ in FSCs, i.e.
when the decoder sends limited-rate instantaneous coded noise-free feedback on
the backward link. Finally, we propose an extension of the 
Blahut-Arimoto algorithm for evaluating 
the capacity when actions can be cost constrained, and demonstrate its application on a few examples.
\end{abstract}


\section{Introduction}
\label{intro}
\par Feedback plays a very important role in communication systems.
Despite proving a pessimistic result in \cite{ShannonFeedback} that feedback
does not increase the capacity of a memoryless channel, Shannon did foresee the 
important role of feedback, which he highlighted in the first Shannon
Lecture. Indeed, 
even for memoryless channels, feedback has its merits, such as simple
capacity achieving coding schemes and improved reliability,
\cite{SchalkwijkKailath}, \cite{Schalkwijk}. Feedback is also known to
increase the capacity for multiple-access channels, \cite{Ozarow} and broadcast
channels, \cite{DueckTwoWay},\cite{OzarowLeung}. 
\par
In his book \cite{GallagerBook}, Gallager introduced finite state
channels (FSCs)
as an apt model for a very broad family of channels with memory.
When no feedback is present and the channel is stationary and indecomposable
without ISI,
the capacity was shown by Gallager in \cite{GallagerBook} and by Blackwell,
Breiman
and Thomasian in \cite{BlackwellBreimanThomasian} to be
\bea
C_{NF}=\lim_{N\rightarrow\infty}\frac{1}{N}\max_{P(x^N)}I(X^N;Y^N).
\eea
For the case of no ISI, stationary and indecomposable finite state channels with
time
invariant deterministic feedback, the capacity was shown
in \cite{HaimTsachyGoldsmith} to be,
\bea
C_{FB}=\lim_{N\rightarrow\infty}\frac{1}{N}\max_{Q(x^N\parallel
z^{N-1})}I(X^N\rightarrow Y^N),
\eea
where $Q(x^N\parallel z^{N-1})$ is causal conditioning introduced by Kramer in
\cite{KramerPHD}, 
\cite{KramerDMC} and is defined as,
\bea
Q(x^N\parallel
z^{N-1})&\stackrel{\triangle}{=}&\prod_{i=1}^NQ(x_i|x^{i-1},z^{i-1}).
\eea
Here $Z_i$ is a time-invariant deterministic function of the output $Y_i$.
Subsequent work on
FSCs included the characterization of the capacity
of finite state multiple access
channel in \cite{HaimTsachyChenMAC}. When the channels have
memory, feedback can
increase the
capacity even for single user channels. One such example is the chemical
channel introduced in 1961 by Blackwell in 
\cite{Blackwell}, also referred to as the \textquoteleft trapdoor
channel\textquoteright\ by Ash in \cite{Ash}. The
capacity of this channel without feedback is a long-standing 
open problem with only bounds on it known, such as those established by
Kobayashi
\textit{et al} in 
\cite{KobayashiMorita}, \cite{Kobayashi}.
With feedback, the capacity of the 
trapdoor channel was computed in \cite{HaimCuffRoyTsachy} using dynamic
programming
approach and shown to be strictly higher than the capacity without
feedback. For Gaussian channels with memory, Cover and Pombra in
\cite{CoverPombra}
showed feedback cannot increase the capacity of an additive white gaussian
channel by more than half of a bit. Kim characterized the capacity of a
wide class of stationary Gaussian channels with feedback in
\cite{KimStatGaussian}. 
\par
Directed information, denoted by $I(X^N\rightarrow Y^N)$, was introduced by
Massey
in \cite{Massey},
where he credits it to Marko \cite{Marko}. It was further shown that
directed information equals mutual information for memoryless channels iff
there is no feedback by Massey and Massey in \cite{MasseyConservation}. 
Directed information also appears in the work of Tatikonda \textit{et al},
\cite{TatikondaPHD}, 
\cite{TatikondaMitter}, where there is generalization of work by Verdu and Han
in \cite{VerduHan}
for the case of
channels with feedback. Capacity of some Markovian Channels was computed using
directed information 
by Yang \textit{et al} in \cite{YangKavcicTatikonda} and Chen and Berger in
\cite{ChenBergerFSMC}. 
Tatikonda also
formulated 
the problem of computing capacities of channels with feedback as a Markov
Decision Process 
in \cite{TatikondaMDP}. Zero error capacity was
also
computed using dynamic programming in
\cite{ZhaoHaimZeroError}. Recently,   
interpretations of directed information in gambling, portfolio theory
and estimation have been characterized in \cite{ZhaoHaimKimTsachy}
, \cite{HaimKimTsachyPortfolio} and
\cite{KimHaimTsachy}. The capacity of the compound channel with feedback was
computed in \cite{ShraderHaim}
using directed information. Directed information also appeared in rate
distortion problems, such as 
source coding with feed-forward by Pradhan and Venkataramanan
\cite{RamjiPradhan}, and implicitly in the 
competitive prediction framework of \cite{MerhavTsachy}. 
\par
In \cite{HaimTsachyVendor}, the notion of
\textit{actions} in a source coding context was introduced. Their
setting is a generalization of the Wyner-Ziv
source coding 
with decoder side information problem in \cite{WynerZiv}, where now the decoder
can take actions based on the index
obtained from the encoder 
to affect the formation or availability of side information. In
\cite{TsachyChannel}, the channel coding dual is studied where the transmitter
takes actions that affect the formation of channel states. This framework
captures 
various new coding scenarios which include two stage
recording on a memory with defects, motivated by similar problems in
magnetic recording and computer memories. Kittichokechai \textit{et al} in
\cite{Kittichokechai} studied 
a variant of the problem in \cite{HaimTsachyVendor} and \cite{TsachyChannel},
where
encoder and decoder both have action dependent partial side information.
However, in the source coding
formulation of \cite{HaimTsachyVendor}, attention was 
restricted to the case where the actions are taken by the decoder while in the
channel coding scenario of
\cite{TsachyChannel} and \cite{Kittichokechai}, actions were taken only
by the encoder. Recently, in \cite{HimanshuHaimTsachyProbing}, the channel
coding setting in \cite{TsachyChannel} and \cite{Kittichokechai} was
generalized, to accommodate the case where both the encoder and the decoder
take channel probing actions, with associated costs, to maximize the rate of
reliable communication. This was referred to as the
\textquoteleft \textit{Probing Capacity}\textquoteright. 
\par
In this paper, we introduce the notion of actions in acquisition of noise-free
feedback or its deterministic function for FSCs. The main contribution of this
paper is in characterizing the cost-capacity
trade-off 
when the feedback observed by the encoder is a deterministic function of an
action taken by the encoder, an action taken by the decoder, and the channel
output, when actions are required to satisfy an average cost constraint. More
precisely, the encoder
observes \textquoteleft
\textit{sampled}\textquoteright\ feedback $Z_i=f(A_{e,i},A_{d,i},Y_i)$, where
$f(\cdot)$ is a deterministic function, $Y_i$ is the
channel
output, $A_{e,i}=A_{e,i}(M,Z^{i-1})$ is the action taken by the encoder as a
function of the message
and the past sampled feedback, and $A_{d,i}$ is the action
taken by
the decoder, where we study two scenarios: one where that action is strictly
causal in the channel output, i.e., $A_{d,i}=A_{d,i}(Y^{i-1})$, and one where it
can depend also on the present channel output, i.e., $A_{d,i}=A_{d,i}(Y^{i})$.
The problem is
motivated by practical applications
where acquisition of the feedback may be costly, and either or both the encoder
and decoder influence whether and what from the channel output is to be fed
back.
\par 
The key technique in our achievability result lies in generating 
both actions 
and input symbol code trees, as described in Section
\ref{achievability}.
 With this achievability, we find most of the proof
follows  that in \cite{HaimTsachyGoldsmith}, except for some cases where care
has
to be taken
to 
properly handle cost constraints. This is because the presence of cost
constraints
 results in breaking down of some properties that were used in [9] such as
sub-additivity. The main contribution of our paper is in obtaining a
multi-letter characterization of the capacity for our communication scenario,
involving maximization over directed information. In order to numerically evaluate the capacity when actions are cost constrained, we also propose a Blahut-Arimoto type algorithm, 
\cite{Blahut},\cite{Arimoto}, similar to that proposed in \cite{Naiss_Permuter_BAA},  where the objective was to maximize the multi-letter directed information expression.  Also our characterization of capacity admits
a dynamic optimization formulation that can lead to analytic closed form
capacity expressions for specific channels, similarly as in \cite{HaimCuffRoyTsachy}, \cite{Elishco_Permuter},  though its pursuance has not been the part of this work. 
\par
A
special case of our framework is when only the encoder or the decoder is the one
taking actions. Under this setting, we motivate and compute a special case of
\textit{to
feed or not to feed back}, i.e., 
where actions are binary corresponding to observing the channel output or not
observing it, the cost constraint corresponding to the fraction of channel
output observations allowed, and the channel
states evolve as a markov chain independent of the channel input process. When
only the encoder takes actions, we
derive a single letter lower bound on 
this capacity and show that it is strictly better than the rate achieved by a
naive time sharing scheme between capacity at zero cost (corresponding to
Gallager's capacity for FSCs in \cite{GallagerBook}) and unit cost (corresponding to the
complete noise-free feedback capacity of \cite{HaimTsachyGoldsmith}). In contrast to this analytical lower bound, our algorithm (cf. BAA-Action, Section \ref{algorithm}), provides a series of upper and lower bounds which converge to the actual capacity (when it exists). For the same
FSC, we also derive bounds on the capacity when only the decoder takes binary actions. A
special case of the framework when only decoder takes actions is that of coding
on the backward link, where the decoder sends a symbol from the action alphabet
based on the channel outputs observed so far, thus operating at an instantaneous
rate which is log the cardinality of said alphabet. The capacity for this case
is characterized in single letter form for some Markovian Channels.
\par 
The rest of paper is organized as follows. Section \ref{problem} describes
the channel model and formulates the
problem studied in this paper. The main results of this paper are outlined in
Section \ref{results}. Section \ref{achievability}
is dedicated to capacity-achieving coding schemes,  while converse results are
proved
in Section \ref{converse}. 
Section \ref{noISI} characterizes the capacity for stationary, indecomposable,
finite state channels without
intersymbol interference (ISI). Section \ref{causaldecoding} generalizes the
framework from decoder taking actions strictly causally dependent on the channel
output (i.e. $A_{d,i}=A_{d,i}(Y^{i-1})$) to the case when decoder can also use
the current output to generate its actions (i.e. $A_{d,i}=A_{d,i}(Y^i)$). As
special cases, Section
\ref{encoderaction} outlines 
the capacity results when actions are taken only by the encoder while the case
when
only decoder takes actions is discussed in Section \ref{decoderaction}.
Section \ref{example1} presents single letter lower 
bounds for a specific 
example of \textit{to feed or not to feed back} (i.e. when actions are binary)
for
Markovian channels when only one of the two, encoder or decoder, takes the 
actions.
Section \ref{example2} establishes that coding on the backward link for FSCs is
a
special case of our general framework, and computes the capacity for an example
of a Markovian channel. Section \ref{algorithm} presents a Blahut-Arimoto type
algorithm (BAA-Action) to find series of converging upper and lower bounds for
the case when encoder take actions which are cost constrained.  The paper is
summarized and concluded in Section
\ref{conclusion}.

\section{Channel Model and Problem Formulation}
\label{problem}
We begin by introducing the notation used throughout this paper.
Let
upper case, lower case, and calligraphic letters denote, respectively, random
variables, specific or deterministic values they may assume,
and
their alphabets. For two jointly distributed random variables, $X$ and $Y$, let
$P_X$, $P_{XY}$ and $P_{X|Y}$ respectively denote the marginal of $X$, joint
distribution of $(X,Y)$ and conditional distribution of 
$X$ given $Y$. $X_{m}^{n}$ is a shorthand for $n-m+1$ tuple
$\{X_m,X_{m+1},\cdots,X_{n-1},X_n\}$. $X^n$ will also denote $X_1^n$. When
$i\leq 0$, $X^i$ denotes null string
as it is also for $X_i^j$, when $i\geq j$. $X^{n\backslash i}$ denotes
$\{X_1,\cdots,X_{i-1},X_{i+1},\cdots,X_n\}$. The cardinality of an alphabet
$\mathcal{X}$ is denoted 
by $\card{\mathcal{X}}$. We impose the assumption of finiteness of cardinality
on all alphabets, unless otherwise indicated.\\
We use the \textit{Causal Conditioning} notation ($\cdot\parallel\cdot$) as
introduced by Kramer in \cite{KramerPHD} and \cite{KramerDMC} :
\bea
P(y^N\parallel x^N)\stackrel{\triangle}{=}\prod_{i=1}^{N}P(y_i|x^i,y^{i-1}).
\eea
We also use the following notation as introduced in \cite{HaimTsachyGoldsmith} :
\bea
P(y^N\parallel
x^{N-1})\stackrel{\triangle}{=}\prod_{i=1}^{N}P(y_i|x^{i-1},y^{i-1}).
\eea
Note that both \textit{causal conditioning}, $P(y^N\parallel x^N)$ and
$P(y^N\parallel x^{N-1})$ are \textit{distributions} on $Y^n$ for a fixed
$x^N$, as they are non negative
for all $x^N,y^N$ and they sum to unity,
 i.e.,
\bea
\sum_{y^N}P(y^N\parallel x^N)=\sum_{y^N}P(y^N\parallel x^{N-1})=1.
\eea
The directed information $I(X^N\rightarrow Y^N)$, as defined by Massey in
\cite{Massey}, is given by,
\bea
I(X^N\rightarrow Y^N)=\sum_{i=1}^{N}I(X^i;Y_i|Y^{i-1})=\E\left[\log
\frac{P(Y^N\parallel X^N)}{P(Y^N)}\right],
\eea
where $\E$ stands for expectation. Naturally, the directed
information conditioned on a random object $S$, $I(X^N\rightarrow Y^N|S)$, is
defined as,
\bea
I(X^N\rightarrow
Y^N|S)\stackrel{\triangle}{=}\sum_{i=1}^{N}I(X^i;Y_i|Y^{i-1},S).
\eea
We model discrete time channels with memory as
Finite State Channels (FSCs) introduced by Gallager in his book
\cite{GallagerBook}, as an
apt class of models for channels with 
memory, e.g. channels with ISI, etc. The channel input symbols take values
in the finite alphabet $\mathcal{X}$ and output denoted by $Y$ takes values in
finite alphabet 
$\mathcal{Y}$.  The state takes
values in a finite alphabet $\mathcal{S}$. The stationary
channel is characterized by the conditional probability law
$P(y_i,s_i|x_i,s_{i-1})$ satisfying,
\bea
P(y_i,s_i|x^i,s^{i-1},y^{i-1},m)=P(y_i,s_i|x_i,s_{i-1}),
\eea
and by the probability of the initial state $P(s_0)$. More precisely, without
loss
of generality, we can make the following assumption
on our channel model,
\bea\label{channelmodel}
P(y_i,s_i|x^i,s^{i-1},y^{i-1},a_{e}^i,a_d^i,m)=P(y_i,s_i|x_i,s_{i-1}),
\eea
where $a_{e,i}\in\mathcal{A}_e$ and $a_{d,i}\in\mathcal{A}_d$ are the encoder
and
decoder actions respectively
as will be explained later. Messages $M\in\mathcal{M}$ are assumed to
be independent of initial state, $s_0$. The FSC is without intersymbol
interference (ISI) if
\bea
P(s_i|s_{i-1},x_{i})=P(s_i|s_{i-1}),
\eea
i.e., the evolution of the channel states is independent of the channel input
process. The basic framework in this paper is the setting depicted in Fig.
\ref{psetup1}.
\begin{figure}[htbp] 
\begin{center}
\scalebox{0.75}{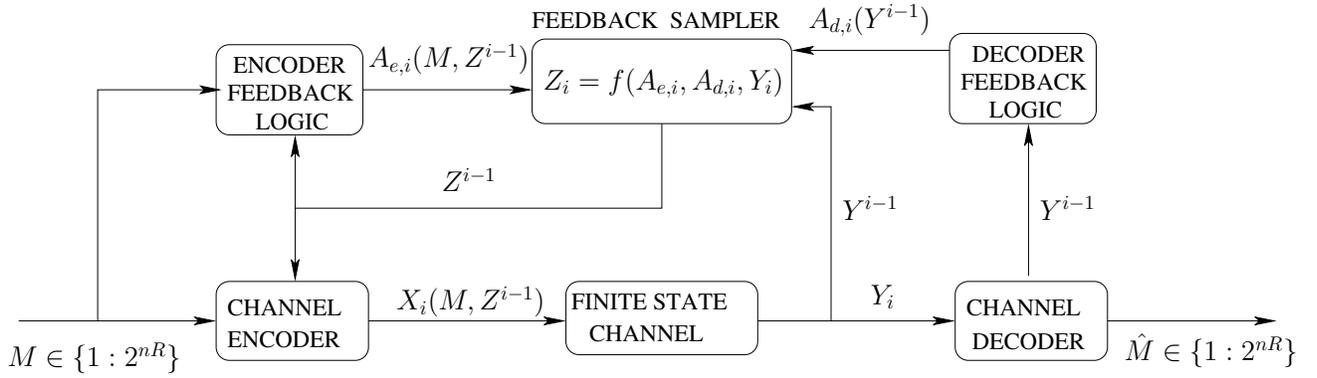}
\caption{Modeling \textbf{Feedback Sampling} for the acquisition of feedback in
Finite State Channels (FSCs).} 
\label{psetup1}
\end{center}
\end{figure}
The communication system has the following building blocks :
\begin{itemize}
 \item \textit{Encoder Feedback Logic} : Generates encoder actions,
 $\{A_{e,i}\}_{i=1}^{N}$, using the function $f_{A_{e,i}} :
\mathcal{M}\times\mathcal {Z}^{i-1}\rightarrow \mathcal{A}_e$ i.e.,
$A_{e,i}=f_{A_{e,i}}(M,Z^{i-1})$, where $Z_i\in\mathcal{Z}$ is the sampled
feedback component. 
 \item \textit{Decoder Feedback Logic} : Generates decoder actions,
 $\{A_{d,i}\}_{i=1}^{N}$, using the function $f_{A_{d,i}} : \mathcal
{Y}^{i-1}\rightarrow \mathcal{A}_d$ i.e.,
$A_{d,i}=f_{A_{d,i}}(Y^{i-1})$. where $Y_i\in\mathcal{Y}$ is the channel
output. 
 \item \textit{Feedback Sampler} : Generates \textit{sampled} feedback, 
$Z_i=f(A_{e,i},A_{d,i},Y_i)$, where $f$ is a deterministic function. 
\item \textit{Channel Encoder} : Constructs channel input symbol,
$X_i(M,Z^{i-1})$, using the encoding
function, $f_{e,i}:\mathcal{M}\times\mathcal{Z}^{i-1}\rightarrow\mathcal{X}$.
\item \textit{Channel Decoder} : Generate the best estimate of the message
given the
channel output, $\hat{M}(Y^N)$, using the decoding function,
$f_d:\mathcal{Y}^N\rightarrow\mathcal{M}$.
\end{itemize}
We are interested in characterizing the maximal rate of reliable communication
under the average cost constraint, 
\bea
\label{cost}
\E\left[\Lambda(A_e^N,A_d^N)\right]=\E\left[\frac{1}{N}\sum_{i=1}^{N}
\Lambda(A_{e,i},A_{d,i})\right ]
\le\Gamma,
\eea
where $\Lambda(\cdot,\cdot)$ is a given cost function satisfying 
$\max_{a_e\in\mathcal{A}_e,a_d\in\mathcal{A}_d}\Lambda(a_e,a_d)=\Lambda_{\max}
<\infty$. 
\par
The joint probability distribution induced by a given scheme,
\bea\label{mainjoint}
&&P_{M,A_{e}^N,A_d^{N},Z^N,X^N,S_0^N,Y^N,\hat{M}}(m,a_{e}^N,a_d^{N},z^N,x^N,
s_0^N ,
y^N , \hat{m})\nonumber\\
&=&\frac{1}{\card{\mathcal{M}}} P_S(s_0)\prod_{i=1}^{n}
\1_{\{a_{d,i}=f_{A_{d,i}}(y^{i-1)}\}}
\1_{\{a_{e,i}=f_{A_{e,i}}(m,z^{i-1})\}}\nonumber\\
&&\times\prod_{i=1}^{n}\1_{\{x_i=f_{e,i}(m,z^{i-1})\}}P(y_i,s_i|x_i,
s_{i-1})\1_{\{z_{i}=f(a_{e,i},a_{d,i},y_i)\}}\times\1_{\{\hat{m}=f_d(y^n)\}}
.\nonumber\\
\eea
\begin{definition}
A rate $R$ is said to be $achievable$ if there exists a sequence
of block codes $(N,\lceil2^{NR}\rceil)$ satisfying (\ref{cost})
such that the maximal probability of error,
\beas
\max_{m\in\{1,\cdots,\lceil2^{NR}\rceil\}}\Pr(\hat{m}\neq
m|\mbox{message $m$ was sent}),
\eeas
vanishes as $N\rightarrow \infty$. The capacity of such a system
is denoted by $C$ which is the supremum of all achievable rates.
\end{definition}

\section{Main Results}
\label{results}
Let $s_0$ denote the initial state. We define $\underline{C}_N(\Gamma)$ and
$\overline{C}_N(\Gamma)$ as,
\bea
\underline{C}_N(\Gamma)&\stackrel{\triangle}{=}&{\frac{1}{N}\max\min_{s_0
}I(X^N\rightarrow
Y^N|s_0)}\\
\overline{C}_N(\Gamma)&\stackrel{\triangle}{=}&{\frac{1}{N}\max\max_{s_0
}I(X^N\rightarrow
Y^N|s_0)}.
\eea
Here \textbf{max} denotes maximization over the
joint probability distribution,
\bea
P(s_0,x^N,a_e^N,a_d^N,y^N,z^N)=P(s_0)Q(x^N,a_e^N\parallel
z^{N-1})Q(a_d^N\parallel y^{N-1})P(y^N\parallel x^N,
s_0)\prod_{i=1}^{N}\1_{\{z_i=f(a_{e,i},a_{d,i},y_i)\}},
\eea
such that $\E[\Lambda(A_e^N,A_d^N)]\le\Gamma$, where 
\bea
Q(x^N,a_e^N\parallel
z^{N-1})&=&\prod_{i=1}^{N}Q(x_i,a_{e,i}|x^{i-1},a^{i-1}_e,z^{i-1})\\
Q(a_d^N\parallel
y^{N-1})&=&\prod_{i=1}^{N}Q(a_{d,i}|a^{i-1}_d,y^{i-1})\\
I(X^N\rightarrow
Y^N|s_0)&=&\sum_{i=1}^{N}I(X^i;Y_i|Y^{i-1},s_0)\nonumber\\
&=&\E\left[\log \frac{P(Y^N\parallel
X^N,s_0)}{P(Y^N|s_0)}\right]\\
P(Y^N\parallel
X^N,s_0)&=&\prod_{i=1}^{N}P(y_i|x^i,y^{i-1},s_0).
\eea
As $z^N$ is a deterministic function of $(a^N_e,a^N_d,y^N)$, from now on we will
consider maximization over the joint probability distribution,
\bea
P(s_0,x^N,a_e^N,a_d^N,y^N)=P(s_0)Q(x^N,a_e^N\parallel
z^{N-1})Q(a_d^N\parallel y^{N-1})P(y^N\parallel x^N,
s_0),
\eea
where $z_i$ will stand for
$f(a_{e,i},a_{d,i},y_i)$ unless otherwise stated. 
Note that effectively maximization in definition of $\underline{C}_N(\Gamma)$
and
$\overline{C}_N(\Gamma)$ is over $Q(x^N,a_e^N\parallel
z^{N-1})Q(a_d^N\parallel y^{N-1})$ as $P(s_0)$ is fixed and $P(y^N\parallel
x^N,s_0)$ (and likewise $P(y^N\parallel x^N)$) is a characteristic of the
channel
given by (Lemma 6 of \cite{HaimTsachyGoldsmith}),
\bea
P(y^N\parallel x^N,s_0)&=&\sum_{s_1^N}\prod_{i=1}^{N}P(y_i,s_i|x_i,s_{i-1})\\
P(y^N\parallel
x^N)&=&\sum_{s_0}P(s_0)P(y^N\parallel
x^N,s_0)=\sum_{s_0^N}P(s_0)\left(\prod_{i=1}^{N}P(y_i,s_i|x_i,s_{i-1}
)\right)\label{channelmodel1}.
\eea
Our main results are as follows,
\begin{itemize}
\item \textbf{Achievable Rate :} For a communication abstraction as in Fig.
\ref{psetup1}, any rate $R$ is achievable such
that,
\bea
R<\lim_{N\rightarrow\infty}\underline{C}_N(\Gamma)=\sup_N\left[\underline{C}
_N(\Gamma)-\frac {
\log
|S|}{N}\right].
\eea
\item \textbf{Converse :} Consider a coding scheme with rate $R$
which achieves reliable communication over the FSC with feedback
sampling as in Fig. \ref{psetup1}. This implies the existence of
$(N,\lceil2^{NR}\rceil)$ codes such that the probability of error $P_{e}^N$ goes
to zero as $N\rightarrow\infty$. For such a  scheme given $\epsilon>0$,
$\exists$
block length $N_0$ such that for all block lengths $N>N_0$ we
have
\bea
R\le\overline{C}_N(\Gamma)+\epsilon.
\eea
\item \textbf{Capacity :} In the following cases we characterize
the capacity exactly,
\begin{enumerate}
\item For an FSC where the probability of the initial state is
positive for all $s_0\in\mathcal{S}$, the capacity is evaluated
exactly,
\bea
C(\Gamma)=\lim_{N\rightarrow\infty}\underline{C}_N(\Gamma).
\eea
\item For stationary \textquoteleft\textit{indecomposable}\textquoteright
\ channels without
ISI with feedback sampling as in Fig. \ref{psetup1}, the capacity
is,
\bea
C(\Gamma)=\lim_{N\rightarrow\infty}{\frac{1}{N}\max I(X^N\rightarrow
Y^N)},
\eea
where \textbf{max} denotes maximization over the joint probability
distribution, 
\bea
P(x^N,a_e^N,a_d^N,y^N)=Q(x^N,a_e^N\parallel
z^{N-1})Q(a_d^N\parallel y^{N-1})P(y^N\parallel x^N),
\eea 
such that $\E[\Lambda(A_e^N,A_d^N)]\le\Gamma$.
\end{enumerate}
\end{itemize}

%
\section{Achievability}
\label{achievability}
We begin this section by proving that the limit of the sequence
$\underline{C}_N(\Gamma)$ exists. We then explain the encoding and decoding
scheme followed by analysis of probability of error, showing that any rate $R$
is achievable such that,
$R<\underline{C}(\Gamma)=\lim_{N\rightarrow\infty}\underline{C}_N(\Gamma)$.
Encoding uses random code-tree generation while decoding uses maximum
likelihood decoding as in \cite{GallagerBook}.

\subsection{Existence of $\underline{C}(\Gamma)$}
\label{existence_lower_bound}
By the following theorem, we prove the existence of the limit of
the sequence $\underline{C}_N(\Gamma)$.
\begin{theorem}
\label{theorem1}
For a finite state channel with $\card{\mathcal{S}}$
states, $\lim_{N\rightarrow\infty}\underline{C}_N(\Gamma)$ exists and,
\bea
\lim_{N\rightarrow\infty}\underline{C}_N(\Gamma)=\sup_{N}\left[\underline{C}
_N(\Gamma)
-\frac{\log\card{\mathcal{S}}}{N}\right].
\eea
\end{theorem}
\begin{proof}
Let $N=n+l$, $n,l\in\mathbf{Z}^+$. Note that from the Section
\ref{encoding} we will show that we achieve $\underline{C}_N$ by
using random coding with distribution of form
$Q(x^N,a_e^N\parallel
z^{N-1})Q(a_d^N\parallel y^{N-1})$ satisfying the
cost constraints.
Let us assume that $\underline{C}_n(\Gamma)$ and $\underline{C}_l(\Gamma)$ are
achieved by $Q(x^n,a_e^n \parallel
z^{n-1})Q(a_d^n \parallel y^{n-1})$ and $Q(x^l,a_e^l\parallel
z^{l-1})Q(a_d^l \parallel y^{l-1})$ respectively. \newline
Consider 
\bea\label{inputdist}
Q(x^N,a_e^N\parallel z^{N-1})&=&Q(x^n,a_e^n\parallel
z^{n-1})Q(x^l,a_e^l\parallel z^{l-1})\\
Q(a_d^N\parallel y^{N-1})&=&Q(a_d^n \parallel y^{n-1})Q(a_d^l \parallel
y^{l-1}).
\eea
Therefore
\bea
\E\left[\Lambda(A_e^N,A_d^N)\right]&=&\frac{n}{N}\E\left[\Lambda(A_e^n,
A_d^n)\right] +
\frac{l}{N}\E\left[\Lambda(A_{e,n+1}^{n+l},A_{d,n+1}^{n+l})\right]\\
&\le&\frac{n\Gamma+l\Gamma}{N}=\Gamma.
\eea
Hence $Q(x^N,a_e^N\parallel z^{N-1})Q(a_d^N\parallel y^{N-1})$ (which is a
distribution) satisfies the
cost requirements, but it may
not be capacity
achieving for blocklength $N$ so,
\bea
N\underline{C}_N(\Gamma)&\ge&\min_{s_0}I(X^N\rightarrow
Y^N|s_0).
\eea
We now follow the steps as in Proof of Theorem 8 in \cite{HaimTsachyGoldsmith}
to arrive at
\bea
N\left[\underline{C}_N(\Gamma)-\frac{\log\card{\mathcal{S}}}{N}\right]
&\ge&n\left[
\underline{C}_n(\Gamma)-\frac{\log\card{\mathcal{S}}}{n}\right]+l\left[
\underline{C}
_l(\Gamma)-\frac{\log\card{\mathcal{S}}}{l}\right].
\eea
Hence the sequence, $\underline{C}_n(\Gamma)$ is super additive for all
$n\in\mathbf{Z}^+$. The theorem is finally proved using the convergence of
super additive sequences, as is done in Theorem 4.6.1
\cite{GallagerBook}.
\end{proof}

\subsection{Encoding Scheme}
\label{encoding} 
\begin{figure}[htbp] 
\begin{center}
\scalebox{0.6}{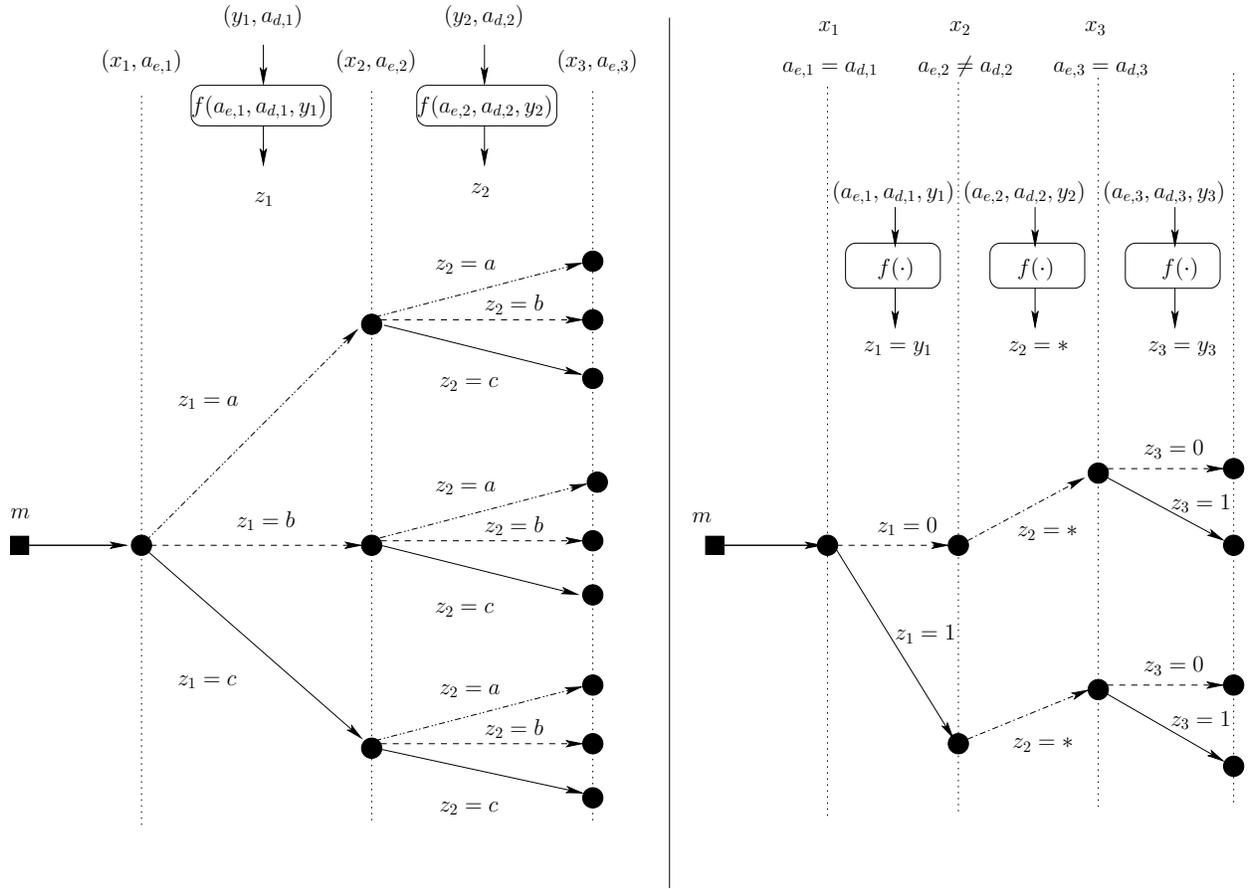}
\caption{This figure illustrates \textit{Encoder Code-Trees} in our coding
scheme. The left hand side figure depicts a general setting where
$\mathcal{Z}=\{a,b,c\}$, and
$z_i=f(a_{e,i},a_{d,i},y_i)$. The tree is shown for $N=3$. The right hand side
shows a
specific example where $a_{d,i}=0$ $\forall i$ and 
output is binary. Actions of encoder, $a_{e,i}\in\{0,1\}$ and 
$z_i=f(a_{e,i},a_{d,i},y_i)=y_i$ if $a_{e,i}= a_{d,i}$ or $a_{e,i}=0$, else
it is
erasure($=\ast$). Hence some portion of the 
tree collapses as by knowing $a_{e,i}$ we know the possible values of
$z_i$, for e.g. $a_{e,i}=1$ implies $z_i=\ast$ and $a_{e,i}=0$ implies,
$z_i=0$ or $1$.} 
\label{codetree}
\end{center}
\end{figure}
Encoding is based on generating separate code trees which is described below.
These are then
revealed to the encoder and the decoder.
\begin{itemize}
 \item \textit{Encoder Code-Tree} : $2^{NR}$ \textbf{code-trees} are
generated as follows, the $i^{th}$
encoder action and
channel input
symbol is
generated using a probability mass function which depends on
previous encoder action and channel input symbols and on the past sampled
feedback
sequence, i.e. $Q(x^i,a_e^i|x^{i-1},a_e^{i-1},z^{i-1})$. 
\item \textit{Decoder Action Code-Tree} : We generate a single code tree at
random, where
the vertex 
represents decoder action symbol, $a_{d,i}$ generated with distribution
$Q(a_{d,i}|a_d^{i-1},y^{i-1})$. Thus the present decoder action depend on the
past actions as well as the past channel output.
\end{itemize}
Note that $\{Q(x^i,a_e^i|x^{i-1},a_e^{i-1},z^{i-1})\}_{i=1}^N$ and 
$\{Q(a_{d,i}|a_d^{i-1},y^{i-1})\}_{i=1}^{N}$ correspond to the joint
distribution on $(X^N,A_e^N,A_d^N,S^N,Y^N)$ such that
constraint $\E\left[\Lambda(A_e^N,A_d^N)\right]\le \Gamma$ is
satisfied.
\par
Fig. \ref{codetree} illustrates the
\textit{Encoder Code-Tree} for a specific
example. The setting in the right in the figure is the illustration of
 the setting of \textit{to
feed or not to feed back}, when the output alphabet is binary and,
\bea
Z_i&=&f(A_{e,i},A_{d,i},Y_i)=\ast\mbox{ if }A_{e,i}\neq A_{d,i}.\\
Z_i&=&f(A_{e,i},A_{d,i},Y_i)=Y_i\mbox{ if }A_{e,i}=A_{d,i},
\eea
where $\ast$ stands for erasure or no feedback. Knowing past channel outputs,
decoder
uses \textit{Decoder Action Code-Tree} to figure
out the
decoder action symbol. Using the decoder action symbol $a_{d,i}$, along with
encoder actions, $a_{e,i}$ and channel output $y_i$, feedback sampler produces
sampled feedback as $z_i=f(a_{e,i},a_{d,i},y_i)$. In this way, given a
message $m$, and the complete sampled feedback
sequence $z^{N-1}$ thus obtained,
there is a particular $(x^{N},a_e^{N})$ which can be found from
the collection of \textit{encoder code trees}. The encoder
thus sends the corresponding $x^N$ though the channel. \textit{Note} that our
coding
scheme is similar in spirit to the code tree generation scheme as in 
\cite{HaimTsachyGoldsmith}. However, here we generate both the cost constrained
encoder actions and channel input symbols in one tree while decoder actions
are generated in
another tree.
\par
By the above code tree generation, we have in our achievability
scheme,
\bea
P(x_i,a_{e,i}|x^{i-1},a_e^{i-1},a_d^{i-1},y^{i-1},s_0^i)&=&P(x_i,a_{e,i}|x^{i-1}
,
a_e^ { i-1},a_d^{i-1},y^{i-1},z^{i-1},s_0^i)\\
&=&Q(x_i,a_{e,i}|x^{
i-1 } , a_e^ { i-1 }, z^ { i-1 } ),\label{condition1}
\eea
where first equality follows from the fact, $z_i=f(a_{e,i},a_{d,i},y_i)$, while
the second equality is due to our coding scheme where the $i^{th}$ input and
encoder action symbol only depend on past input symbols, actions and sampled
feedback. Similarly since $i^{th}$ decoder action only depends on
past decoder actions and channel output, we have,
\bea\label{condition2}
P(a_{d,i}|a_d^{i-1},y^{i-1},x^i,a_e^i,s_0^i)=Q(a_{d,i}|a_d^{i-1},y^{i-1}).
\eea
\begin{lemma}
\label{achjoint}
\label{lemma1}
 The joint probability distribution on $(s_0,x^N,a_e^N,a_d^N,y^N)$, by the
achievability scheme described above is,
\bea
P(s_0,x^N,a_e^N,a_d^N,y^N)=P(s_0)Q(x^N,a_e^N\parallel z^{N-1})Q(a_d^N\parallel
y^{N-1})P(y^N\parallel x^N,s_0).
\eea
\end{lemma}
\begin{proof}
Using Property 1 in Appendix \ref{basic}, we have, 
\bea
\label{one}
P(s_0,x^N,a_e^N,a_d^N,y^N)=P(s_0)Q(x^N,a_e^N,a_d^N\parallel
y^{N-1},s_0)P(y^N\parallel
x^N,a_e^N,a_d^N,s_0).
\eea
From definition of causal conditioning and using Eq. (\ref{condition1}) and
(\ref{condition2}) we have,
\bea\label{two}
Q(x^N,a_e^N,a_d^N\parallel
y^{N-1},s_0)=Q(x^N,a_e^N\parallel
z^{N-1})Q(a_d^N\parallel y^{N-1}).
\eea
Now again using Eq. (\ref{condition1}) and (\ref{condition2}) and the
channel model assumption in Eq. (\ref{channelmodel}) consider,
\bea
P(s_0^N,x^N,a_e^N,a_d^N,y^N)&=&P(s_0)Q(x^N,a_e^N\parallel
z^{N-1})Q(a_d^N\parallel y^{N-1})\prod_{i=1}^{N}P(y_i,s_i|x_i,s_{i-1}),
\eea
Summing over, $s_1^N$ and using the characterization of
$P(y^N\parallel x^N,s_0)$ in Section \ref{results} as,
\bea
P(y^N\parallel x^N,s_0)=\sum_{s_1^N}\prod_{i=1}^{N}P(y_i,s_i|x_i,s_{i-1}),
\eea
we obtain,
\bea
P(s_0,x^N,a_e^N,a_d^N,y^N)&=&\sum_{s_1^N}P(s_0^N,x^N,a_e^N,a_d^N,y^N)\\
&=&P(s_0)Q(x^N,a_e^N\parallel
z^{N-1})Q(a_d^N\parallel y^{N-1})P(y^N\parallel x^N,s_0).
\eea
\end{proof}
\begin{corollary}
From the steps in previous lemma it immediately implies,
\bea\label{three}
P(y^N\parallel
x^N,a_e^N,a_d^N,s_0)=P(y^N\parallel
x^N,s_0).
\eea
Note that likewise it can be also shown as in Eq.
(\ref{three}) that,
\bea\label{four}
P(y^N\parallel
x^N,a_e^N,a_d^N)=P(y^N\parallel
x^N),
\eea
which we will use in next section on decoding.
\end{corollary}

\subsection{Decoding}
\label{decoding}
The decoder performs ML
decoding, i.e. it chooses the message $m$ for which
$P(y^N|m)$ is maximized.
\bea
P(y^N|m)&=&\prod_{i=1}^{N}P(y_i|y^{i-1},m)\\
&\stackrel{(a)}{=}&\prod_{i=1}^{N}P(y_i|y^{i-1},a_d^i(y^{i-1}),m,x^i(m,z^{i-1}),
a_e^i(m.z^{i-1}))\\
&\stackrel{(b)}{=}&\prod_{i=1}^{N}P(y_i|y^{i-1},a_d^i(y^{i-1}),x^i(m,z^{i-1}),
a_e^i(m.z^{i-1}))\\
&\stackrel{}{=}&P(y^N\parallel
x^N,a_e^N,a_d^N)\label{eq11}\\
&\stackrel{(c)}{=}&P(y^N\parallel x^N)\label{eq12},
\eea
where (a) follows from the fact that knowing $m$ and $y^{i-1}$,
we know $(x^i,a_e^i,a_d^i)$. This can be iteratively shown. Given $m$ we
know $(x_1(m),a_1(m))$. We also know $a_{d,1}$. Given $y_1$,
$z_1=f(a_{e,1},a_{d,1},y_1)$.
Hence now we know, $(x_2(m,z_1),a_{e,2}(m,z_1),a_{d,2}(y_1))$. Iteratively we
can
conclude that for a given message $m$ and true feedback
sequence, $y^{i-1}$, we can construct $(x^i,a_e^i,a_d^i)$ knowing the codebooks.
(b) follows from the assumption on channel model in Eq. (\ref{channelmodel}) and
(c) follows from Eq. (\ref{four}).
Hence ML decoding to construct message estimate, $\hat{m}$ 
can also be done my maximizing causal conditioning, i.e.,
\bea
\hat{m}=\argmax_{m}P(y^N|m)=\argmax_{x^N}P(y^N\parallel x^N).
\eea

\subsection{Calculation of Probability of Error}
\label{errorprob}
We will see in this section that most of the proofs are similar to that
in \cite{HaimTsachyGoldsmith} with $Q(x^N,a^N\parallel z^{N-1})$ being
replaced with $Q(x^N,a_e^N\parallel z^{N-1})Q(a_d^N\parallel y^{N-1})$. This is 
justifiable from our coding scheme that uses a distribution which is causal
conditioning,
$Q(x^N,a_e^N\parallel z^{N-1})Q(a_d^N\parallel y^{N-1})$ and the optimal
decoding which is finding 
\bea
\argmax_{m}P(y^N|m)=\argmax_{x^N}P(y^N\parallel x^N).
\eea
Also from Lemma \ref{achjoint} we have,
\bea
P(x^N,a_e^N,a_d^N, y^N) &=&\sum_{s_0}P(s_0,x^N,a_e^N,a_d^N, y^N)\\
&=&Q(x^N,a_e^N\parallel
z^{N-1})Q(a_d^N\parallel
y^{N-1})\sum_{s_0}P(s_0)P(y^N\parallel
x^N,s_0)\\
&=&Q(x^N,a_e^N\parallel
z^{N-1})Q(a_d^N\parallel
y^{N-1})P(y^N\parallel
x^N),
\label{factor}
\eea
where last equality follows from the characterization of $P(y^N\parallel x^N)$
in Section \ref{results}. Due to factorization in Eq. (\ref{factor}) similar to
the one in
\cite{HaimTsachyGoldsmith}
as
\bea
P(x^N,y^N)=Q(x^N\parallel z^{N-1})P(y^N\parallel x^N),
\eea
we have parallelism in the proofs.
\par
Note that from now on we will not state the condition of cost constraints,
i.e., $\E[\Lambda(A_e^N,A_d^N)]\le\Gamma$ explicitly in maximizing
distribution. The distribution $Q(x^N,a_e^N\parallel z^{N-1})Q(a_d^N\parallel
y^{N-1})$ will be assumed to be the one satisfying cost constraints.
Let $P_{e,m}$ denote the probability of error of ML decoding
when message $m$ was sent. Given message $m$, $Y_m^c$ denotes
the set of outputs that cause error in decoding $m$, i.e.,
\bea
P_{e,m}=\sum_{y^N\in Y_m^c}P(y^N|m).
\eea
\begin{theorem}
\label{theorem2}
Let $M$ denote the total number of messages used in transmission
and $\mathsf{E}(P_{e,m})$ denote the average probability of error
over these ensemble of codes. Then for any $\rho$, $0<\rho\le1$,
\bea
\E(P_{e,m})\le
(M-1)^\rho\sum_{y^N}\left[\sum_{x^N,a_e^N,a_d^N}Q(x^N,a_e^N\parallel
z^{N-1})Q(a_d^N\parallel y^{N-1})P(y^N\parallel
x^N)^{\frac{1}{1+\rho}}\right]^{1+\rho}.\label{errorbound}
\eea
\end{theorem}
\begin{proof}Refer to Appendix \ref{proof2}.
\end{proof}
Let $P_{e,m}(s_0)$ denote the probability of error given the
initial state of FSC was $s_0$ and the message $m$ was sent.
\begin{theorem}\label{theorem3}
Consider FSC with feedback sampling (Fig. \ref{psetup1}) having
$\card{\mathcal{S}}$
states. For any positive integer $N$ and any positive rate $R$,
$\exists$ $(N,M)$ code for which for all messages
$m\in\{1,\cdots,\lceil2^{NR}\rceil\}$, all initial states $s_0$
and all $\rho$, $0<\rho\le1$,
\bea
P_{e,m}(s_0)\le 4\card{\mathcal{S}}2^{-N\left[-\rho
R+F_N(\rho)\right]},
\eea
where 
\bea
F_N(\rho)=-\frac{\rho\log\card{\mathcal{S}}}{N}+\max_{Q(x^N,a_e^N\parallel
z^{N-1})Q(a_d^N\parallel y^{N-1})}\left[
\min_{s_0}\E_{o,N}(\rho,Q(x^N,a_e^N\parallel z^{N-1})Q(a_d^N\parallel
y^{N-1}),s_0)\right],
\eea
and
\bea
&&\E_{o,N}(\rho,Q(x^N,a_e^N\parallel z^{N-1})Q(a_d^N\parallel
y^{N-1}),s_0)\nonumber\\&=&-\frac{1}{N}\log\sum_{y^N}\left[\sum_{x^N,a_e^N,a_d^N
} Q(x^N ,
a_e^N\parallel z^{N-1})Q(a_d^N\parallel y^{N-1})P(y^N\parallel
x^N,s_0)^{\frac{1}{1+\rho}}\right]^{1+\rho}.
\eea
\end{theorem}
\begin{proof}Proof is following the steps in proof of Theorem 10 in
\cite{HaimTsachyGoldsmith} 
once we have obtained the bound on $\E(P_{e,m})$ [Eq. (\ref{errorbound})] in
Theorem \ref{theorem2}.
\end{proof}

\begin{theorem}
\label{theorem4}
$\E_{o,N}(\rho,Q(x^N,a_e^N\parallel z^{N-1})Q(a_d^N\parallel y^{N-1}),s_0)$
has
the
following properties,
\bea
\E_{o,N}(\rho,Q(x^N,a_e^N\parallel z^{N-1})Q(a_d^N\parallel
y^{N-1}),s_0)&\ge&0\label{eqth41}\\
\label{eqnI}\frac{1}{N}I(X^N\rightarrow Y^N|s_0))\ge\frac{\partial
\E_{o,N}(\rho,Q(x^N,a_e^N\parallel z^{N-1})Q(a_d^N\parallel
y^{N-1}),s_0)}{\partial\rho}&>&0\nonumber\\\label{eqth42}\\
\frac{\partial^2 \E_{o,N}(\rho,Q(x^N,a_e^N\parallel z^{N-1})Q(a_d^N\parallel
y^{N-1}),s_0)}{\partial\rho^2}&>&0,
\eea
where equality in Eq. (\ref{eqth41}) holds when $\rho=0$, and equality holds on
the left side of Eq. (\ref{eqth42}) when $\rho=0$. 
\end{theorem}
\begin{proof}
Omitted as it is similar to proof of Theorem 11 in
\cite{HaimTsachyGoldsmith} with $Q(x^N\parallel z^{N-1})$ replaced by
$Q(x^N,a_e^N\parallel
z^{N-1})Q(a_d^N\parallel y^{N-1})$.
\end{proof}
\begin{lemma}
\label{lemma2}
We have the following results for the convergence of $F_N(\rho)$,
\bea
\lim_{N\rightarrow\infty}F_N(\rho)=F_{\infty}(\rho)=\sup_{N}F_N(\rho),
\eea
for $0\le\rho\le1$. The convergence of $F_N(\rho)$ is uniform in
$\rho$ and $F_{\infty}(\rho)$ is uniformly continuous for
$\rho\in[0,1]$.
\end{lemma}
\begin{proof}Omitted. Proof similar to Lemma 13 in
\cite{HaimTsachyGoldsmith}.
\end{proof}
\begin{theorem}
\label{theorem5}
For any  FSC with feedback logic let,
\bea
E_r(R)=\max_{0\le\rho\le1}\left[F_{\infty}(\rho)-\rho R\right].
\eea
Then for any $\epsilon>0$, $\exists$ $N(\epsilon)$ such that for
$N\ge N(\epsilon)$, $\exists$ an $(N,M)$ code such that for all
$m,1\le m\le M=\lceil2^{NR}\rceil$, and all initial states,
\bea
P_{e,m}(s_0)\le2^{-N\left[E_r(R)-\epsilon\right]}.
\eea
\end{theorem}
\begin{proof}Proof is similar to Theorem 14 in
\cite{HaimTsachyGoldsmith} using above Theorems 
\ref{theorem1}, \ref{theorem2}, \ref{theorem3}, \ref{theorem4} and Lemma
\ref{lemma2} to conclude that for every $s_0$, there exists a
$\rho^*$ such that $F_{\infty}(\rho^*)-\rho^* R>0$, for all
$R<\underline{C}(\Gamma)$.
\end{proof}

\section{Converse}
\label{converse}
In this section, we will first prove some converse results. Later in this
section, we will show that for FSCs where probability of
initial state is
positive for all $s_0\in\mathcal{S}$, the achievable rate and the upper
bound coincide and hence the capacity is given by
$\underline{C}(\Gamma)$.
\label{converse}
\begin{theorem}
\label{theorem5} 
Consider a coding scheme with rate $R$
which achieves reliable communication over the FSC with feedback
sampling as in Fig. \ref{psetup1} meeting the average cost constraints, Eq.
(\ref{cost}). 
For such a  scheme given any $\epsilon_N>0$,
$\exists$
block length $N_0$ such that for all block lengths $N>N_0$ we
have
\bea
R\le\overline{C}_N(\Gamma)+\epsilon_N.
\eea
\end{theorem}
\begin{proof}
Let a message $m$ is chosen uniformly with probability
$2^{-NR}$.
\bea
NR&=&H(M)\\
&\stackrel{(a)}{=}&H(M|S_0)\\
&\stackrel{}{=}&I(M;Y^N|S_0)+H(M|Y^N,S_0)\\
&\stackrel{}{\le}&I(M;Y^N|S_0)+H(M|Y^N)\\
&\stackrel{(b)}{\le}&I(M;Y^N|S_0)+1+P_e^{(N)}NR\\
&\stackrel{}{=}&\sum_{i=1}^{N}H(Y_i|Y^{i-1},S_0)-H(Y_i|Y^{i-1},X^i,A_e^i,A_d^i,
M,S_0)+1+P_e^ {
(N)}NR\\
&\stackrel{(c)}{=}&\sum_{i=1}^{N}H(Y_i|Y^{i-1},S_0)-H(Y_i|Y^{i-1},X^i,S_0)+1+P_e^{
(N)}NR\\
&\stackrel{}{=}&\sum_{i=1}^{N}I(X^i;Y_i|Y^{i-1},S_0)+1+P_e^{(N)}NR\\
&\stackrel{}{=}&I(X^N\rightarrow
Y^N|S_0)+1+P_e^{(N)}NR\\
&\stackrel{}{\le}&\max_{s_0}I(X^N\rightarrow
Y^N|s_0)+1+P_e^{(N)}NR\\
&\stackrel{(d)}{\le}&\max \max_{s_0}I(X^N\rightarrow
Y^N|s_0)+1+P_e^{(N)}NR\\
&\stackrel{}{=}&N\overline{C}_N(\Gamma)+1+P_e^{(N)}NR,\label{limit3}
\eea
where
\begin{itemize}
\item (a) follows from the independence of message and initial state.
\item (b) follows from Fano's inequality.
\item (c) follows from proof of MC1 in Appendix \ref{markov}.
\item (d) has its first maximization over the joint probability distribution
\bea\label{eq85}
P(s_0,x^N,a_e^N,a_d^N,y^N)=P(s_0)Q(x^N,a_e^N\parallel
z^{N-1})Q(a_d^N\parallel y^{N-1})P(y^N\parallel x^N,s_0),
\eea
which satisfy the expected cost constraints,
$\E[\Lambda(A_e^N,A_d^N)]\le\Gamma$ 
 and Eq. (\ref{eq85}) follows from Lemma \ref{cor3} in
Appendix \ref{markov}.
\end{itemize}
Hence we have for sufficiently large $N$ for any given $\epsilon_N >0$,
\bea
R\le \overline{C}_N(\Gamma)+\epsilon_N.
\eea
\end{proof}
\textit{Note} that unlike in \cite{HaimTsachyGoldsmith}, limit of
$\overline{C}_N(\Gamma)$ may not exist
 because sub-additivity (like the one in Theorem 16  in
\cite{HaimTsachyGoldsmith}) breaks due to the presence of cost constraints.
Hence for a general FSC, we have the above converse result for a give
blocklength
$N$ and probability of error, $\epsilon_N$. However if for the FSC, the
probability of initial state is positive for all states, then we have the exact
capacity as shown by the following theorem.
 \begin{theorem}
\label{theorem]}
Consider an FSC with feedback logic where all the initial states
$\in\mathcal{S}$ have positive probability. The capacity is
$\underline{C}(\Gamma)$.
\end{theorem}
\begin{proof}
The proof is similar to Theorem 17 in \cite{HaimTsachyGoldsmith} with change in
equalities in 
(c), (d) and (e) below. Let
$P^N_e(s_0)$ denote the probability 
of error when the initial state is $s_0$. Since every initial state $s_0\in
\mathcal{S}$ can occur 
with non zero probability, this implies that there exists a sequence of block
codes $(N,\lfloor 2^{NR}\rfloor)$ 
with $P^N_e(s_0)\rightarrow 0$, $\forall s_0\in\mathcal{S}$. Hence we have,
\bea
NR&=&H(M)\\
&\stackrel{(a)}{=}&H(M|s_0)\\
&\stackrel{}{=}&I(M;Y^N|s_0)+H(M|Y^N,s_0)\\
&\stackrel{(b)}{\le}&I(M;Y^N|s_0)+1+P_e^{(N)}(s_0)NR\\
&\stackrel{(c)}{=}&\sum_{i=1}^{N}H(Y_i|Y^{i-1},s_0)-H(Y_i|Y^{i-1},X^i,A_e^i,
A_d^i,M ,
s_0)+1+P_e^{
(N)}(s_0)NR\\
&\stackrel{(d)}{=}&\sum_{i=1}^{N}H(Y_i|Y^{i-1},s_0)-H(Y_i|Y^{i-1},X^i,A_e^i,
A_d^i,
s_0)+1+P_e^{
(N)}(s_0)NR\\
&\stackrel{(e)}{=}&\sum_{i=1}^{N}H(Y_i|Y^{i-1},s_0)-H(Y_i|Y^{i-1},X^i,
s_0)+1+P_e^{
(N)}(s_0)NR\\
&\stackrel{}{=}&\sum_{i=1}^{N}I(X^i;Y_i|Y^{i-1},s_0)+1+P_e^{(N)}(s_0)NR\\
&\stackrel{}{=}&I(X^N\rightarrow
Y^N|s_0)+1+P_e^{(N)}NR\label{five}\\
&\stackrel{(f)}{\le}&\min_{s_0}\left[I(X^N\rightarrow
Y^N|s_0)+1+P_e^{(N)}(s_0)NR\right],
\eea
where
\begin{itemize}
\item (a) follows from the fact that
message $M$ is independent of initial state
$s_0$.
\item (b) follows from Fano's inequality.
\item (c) follow from similar arguments as in
\ref{decoding}.
\item (d) follows from the assumption of channel model as in Eq.
(\ref{channelmodel}).
\item (e) follows from proof of MC1 in Appendix \ref{markov}.
\item (f) follows from the fact that Eq. (\ref{five}) is true for all
$s_0\in\mathcal{S}$.
\end{itemize}
Hence since we have $P^N_e(s_0)\rightarrow 0$, $\forall s_0\in\mathcal{S}$, we
have,
\bea
R\le\lim_{N\rightarrow\infty}\frac{1}{N}\max \min_{s_0}I(X^N\rightarrow Y^N),
\eea
where due to Lemma \ref{cor3} in Appendix \ref{markov} the maximization is over
the 
joint probability distribution, 
\bea
P(s_0,x^N,a_e^N,a_d^N,y^N)=P(s_0)Q(x^N,a_e^N\parallel
z^{N-1})Q(a_d^N\parallel y^{N-1})P(y^N\parallel x^N,s_0),
\eea
which satisfy the expected cost constraints,
$\E[\Lambda(A_e^N,A_d^N)]\le\Gamma$. This implies from the achievability result
of Section \ref{achievability} that
capacity is,
\bea
C(\Gamma)=\underline{C}(\Gamma).
\eea
\end{proof}

\section{Capacity for Stationary Indecomposable FSC without ISI}
\label{noISI}
We assume now
that state transition is a separate markov chain
and does not depend on input, i.e.,
$P(y_i,s_i|s_{i-1},x_i)=P(s_i|s_{i-1})P(y_i|s_i,s_{i-1},x_i)$.
Such a channel is said to have no ISI. We further assume this
channel is \textit{indecomposable} as the definition given
below,
\begin{definition}
An FSC without ISI is said to be indecomposable if, for every
$\epsilon>0$, $\exists N_0$ such that $\forall N>N_0$
\bea
|P(s_N|s_0)-P(s_N|s_0')|\le\epsilon\mbox{ }\forall\mbox{
}s_N,s_0,s_0'\label{part0}.
\eea
\end{definition}
 A necessary and sufficient condition for a no ISI, FSC to be indecomposable
[c.f. Theorem 4.6.3, \cite{GallagerBook}] is that there exists a choice for the
$n^{th}$ state, say $s_n$, such that,
\bea
q(s_n|s_0)>0, \mbox{ } \forall s_0\in\mathcal{S}.
\eea
Furthermore, if the channel is indecomposable,
$n$ above can always be 
taken less than $2^{\card{\mathcal{S}}^2}$.
This condition [Theorem 6.3.2, 
\cite{GallagerBook}] also implies existence of a unique steady-state stationary
distribution $\pi(s)$ , i.e.,
\bea
\lim_{N\rightarrow\infty} P(S_N=s|s_0)=\pi(s)\label{part00}.
\eea
The channel is stationary if $P(s_0)=\pi(s_0)$.
\begin{theorem}
\label{theorem8}
For a stationary and indecomposable FSC without ISI and with
communication abstraction as in Fig. \ref{psetup1}, the capacity of the channel
is given by,
\bea
C(\Gamma)=\lim_{N\rightarrow\infty}C^N(\Gamma)=\lim_{
N\rightarrow\infty } \frac { 1 } { N }
\max I(X^N\rightarrow Y^N),
\eea
where \textbf{max} denotes maximization over $Q(x^N,a_e^N\parallel
z^{N-1})Q(a_d^N\parallel y^{N-1})$ such that
$\E[\Lambda(A_e^N,A_d^N)]\le\Gamma$.
\end{theorem}
\begin{proof}The proof is similar to proof of Theorem 18 in
\cite{HaimTsachyGoldsmith} with $Q(x^N\parallel
z^{N-1})$ replaced by $Q(x^N,a_e^N\parallel
z^{N-1})Q(a_d^N\parallel y^{N-1})$.
\end{proof}

\section{Causal Action Encoding at Decoder}
\label{causaldecoding}
In this section we generalize the framework in Fig. \ref{psetup1}, where now
decoder actions also depend on the current channel output, i.e.,
$A_{d,i}=f_{A_{d,i}}(Y^i)$. The setting is depicted in Fig. \ref{psetup5}. Note
that the capacity in this generalized setting
can be strictly better than that in Fig. \ref{psetup1}. To get an intuition for
it, one can consider a markovian channel, i.e., an FSC for which,
\bea
P(Y_i,S_i|X_i,S_{i-1})=P(Y_i|X_i,S_{i-1})P(S_i|S_{i-1}).
\eea
The decoder knows the states along with the output on the fly, feds back the
effective output, $Y_{FB,i}=(Y_i,S_i)$ to the feedback sampler and the feedback
sampling function is specialized to $f(A_{e,i},A_{d,i},Y_{FB,i})=A_{d,i}$,
$A_{d,i}=f_{A_{d,i}}(Y^i.S^i)$. Further
$\card{\mathcal{A}}=\card{\mathcal{S}}$ and there are no cost constraints. We
will see later in Section \ref{example2} that this is the setting of coding on
 the backward link in FSCs with no constraints on active feedback symbols. As
will
be shown in Section \ref{example2} that the capacity of this system is the same
as that when encoder and decoder both have state information and it is achieved
by setting 
$A_{d,i}=S_{i}$. Here we are able to do better because $X_i(M,A^{i-1})$ can be
generated
using $S_{i-1}$ on which the channel output depends ($P(Y_i|X_i,S_{i-1})$).
Thus, it is easy to see that under such a framework for the setting in Fig.
\ref{psetup1}, i.e., when $A_{d,i}=f_{A_{d,i}}(Y^{i-1},S^{i-1})$, capacity can
be 
comparatively strictly less, as channel input can at most depend on state upto
$S^{i-2}$ and has no information about the state $S_{i-1}$ which determines
the channel output.
\begin{figure}[htbp] 
\begin{center}
\scalebox{0.65}{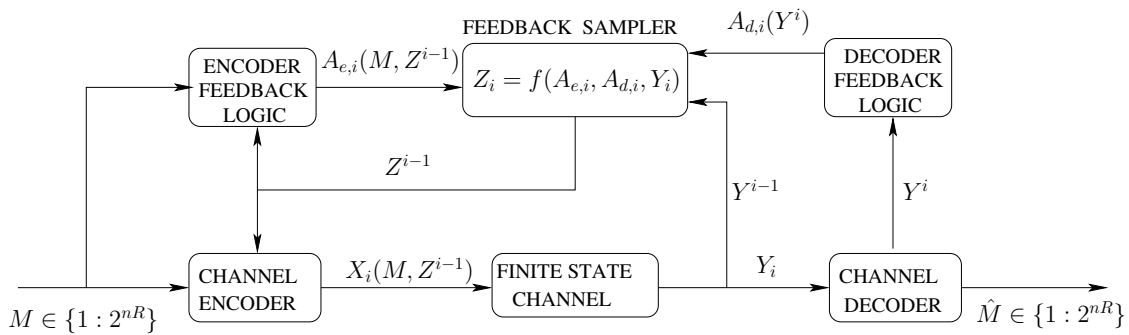}
\caption{Modeling \textbf{Feedback Sampling} for the acquisition of feedback in
Finite State Channels (FSCs) when decoder can use the \textbf{current channel
output} also
to generate actions.}
\label{psetup5}
\end{center}
\end{figure}

\begin{theorem}
\label{theorem12}
 Consider the system in Fig. \ref{psetup5}. We have the following results
paralleling those in Section \ref{results} (for the setting of
Fig. \ref{psetup1}).
\par
Let $s_0$ denotes the initial state. We define
$\underline{C}_{N,causal}(\Gamma)$ and
$\overline{C}_{N,causal}(\Gamma)$ as (where causal indicates that decoder
actions can also depend on current channel output),
\bea
\underline{C}_{N,causal}(\Gamma)&\stackrel{\triangle}{=}&{\frac{1}{N}\max\min_{
s_0
}I(X^N\rightarrow
Y^N|s_0)}\\
\overline{C}_{N,causal}(\Gamma)&\stackrel{\triangle}{=}&{\frac{1}{N}\max\max_{
s_0
}I(X^N\rightarrow
Y^N|s_0)}.
\eea
Here \textbf{max} denotes maximization over the
joint probability distribution,
\bea
P(s_0,x^N,a_e^N,\phi_d^N,y^N,z^N)=P(s_0)Q(x^N,a_e^N\parallel
z^{N-1})Q(\phi_d^N\parallel y^{N-1})P(y^N\parallel x^N,
s_0)\prod_{i=1}^N\1_{\{z_i=f(a_{e,i},\phi_{d,i}|_{y_i},y_i)\}},
\eea
such that $\E[\Lambda(A_e^N,\Phi_d^N|_{Y^N})]\le\Gamma$. Here
$\phi_d^N,\phi_d^N|_{y^N}$ are particular realizations of
random variables $\Phi_d^N, \Phi_d^N|_{Y^N}$ and
\bea
\phi_{d,i}|_{y_i}&=&f_{A_{d,i},y_i}(y^{i-1})\in\mathcal{A}_d\\
\phi_{d,i}&=&\{f_{A_{d,i},y}(y^{i-1}),y\in\mathcal{Y}\}
\in\mathcal{A}_d^{\card{\mathcal{Y}}}\\
\phi_d^N|_{y^N}&=&\{\phi_{d,i}|_{y_i}\}_{i=1}^{N}.
\eea
With slight abuse of notation $\phi_{d,i}|y$ for each $y\in\mathcal{Y}$ denotes
a function from $\mathcal{Y}^{i-1}$ to
$\mathcal{A}_d$ and $\phi_{d,i}$ can be treated as a vector of
functions $\{\phi_{d,i}|y\}_{y\in\mathcal{Y}}$. 
Note that $A_{d,i}=\phi_{d,i}|_{y_i}$ and hence $\{\phi_d^N|_{y^N}\}$ denotes
the decoder action sequence.
\begin{enumerate}
\item \textbf{Achievable Rate :} For a communication abstraction as in Fig.
\ref{psetup5}, any rate $R$ is achievable such
that,
\bea
R<\lim_{N\rightarrow\infty}\underline{C}_{N,causal}(\Gamma)=\sup_N\left[
\underline { C }
_{N,causal}(\Gamma)-\frac {
\log
|S|}{N}\right].
\eea
\item \textbf{Converse :} Consider a coding scheme with rate $R$
which achieves reliable communication over the FSC with feedback
sampling as in Fig. \ref{psetup5}. This implies the existence of
$(N,\lceil2^{NR}\rceil)$ codes such that the probability of error $P_{e}^N$ goes
to zero as $N\rightarrow\infty$. For such a  scheme given $\epsilon>0$,
$\exists$
block length $N_0$ such that for all block lengths $N>N_0$ we
have
\bea
R\le\overline{C}_{N,causal}(\Gamma)+\epsilon.
\eea
\item \textbf{Capacity :} In the following cases we characterize
the capacity exactly,
\begin{enumerate}
\item For an FSC where the probability of the initial state is
positive for all $s_0\in\mathcal{S}$, the capacity is evaluated
exactly,
\bea
C_{causal}(\Gamma)=\lim_{N\rightarrow\infty}\underline{C}_{N,causal}(\Gamma).
\eea
\item For stationary \textquoteleft\textit{indecomposable}\textquoteright
\ channels without
ISI with feedback sampling as in Fig. \ref{psetup5}, the capacity
is,
\bea
C_{causal}(\Gamma)=\lim_{N\rightarrow\infty}{\frac{1}{N}\max I(X^N\rightarrow
Y^N)},
\eea
where maximization is over the joint probability distribution,
\bea
P(x^N,a_e^N,\phi_d^N,y^N,z^N)=Q(x^N,a_e^N\parallel
z^{N-1})Q(\phi_d^N\parallel y^{N-1})P(y^N\parallel
x^N)\prod_{i=1}^N\1_{\{z_i=f(a_{e,i},\phi_{d,i}|_{y_i},y_i)\}}\nonumber,
\eea
such that $\E[\Lambda(A_e^N,\Phi_d^N|_{Y^N})]\le\Gamma$
\end{enumerate}
\end{enumerate}
\end{theorem}
\begin{proof}
 The proof is straightforward as it uses the similar results as
stated in Section \ref{results} for the framework in Fig. \ref{psetup1} where
decoder actions do not depend on current channel output. The argument is as
follows. Notice that the setting in Fig. \ref{psetup5} where decoder takes
actions $A_{d,i}(Y^i)$ and the sampling function is
$Z_i=f(A_{e,i},A_{d,i},Y_i)$ is equivalent to the setting in Fig. \ref{psetup1}
where decoder takes actions
$\tilde{A}_{d,i}(Y^{i-1})\in\tilde{\mathcal{A}}_d=\mathcal{A}_d^{\card{\mathcal{
Y } } }$, or
$\tilde{A}_{d,i}=\Phi_{d,i}$ as defined
in the Theorem above and feedback sampling function is, \bea
Z_i=g(A_{e,i},\tilde{A}_{d,i},Y_i)=f(A_{e,i},\Phi_{d,i}|_{Y_i},Y_i).
\eea 
More precisely operationally, the generalized setting when decoder takes action
which also depend on current output is equivalent when decoder takes an action
vector, depending only on past channel output, for each of the possible Y's, and
the feedback sampling function uses the current output to extract out the
corresponding action from this action vector to generate sampled feedback. Hence
all the above results are derived from those in Section \ref{results} by the
transformation, $A_d\rightarrow\Phi_d$, $f(\cdot)\rightarrow g(\cdot)$. The
cost constraints hence are equivalent to
$\E[\Lambda(A_e^N,\Phi_d^N|_{Y^N})]\le\Gamma$.  The idea is similar to that of Shannon strategies, \cite{ShannonChannel}
\end{proof}
\begin{note}
Note that we started by solving a seemingly more restrictive case, i.e., the
setting in Fig. 1 where decoder actions depend on the channel output
strictly causally. In this section, we applied our results for the
setting of Fig. \ref{psetup1}
to characterize fundamental limits for the setting in Fig. \ref{psetup5}, where
decoder actions can depend also on the current channel
output, by showing that the latter setting can be embedded in the
former via an appropriate extension of the decoder action alphabet.
Thus, the setting of Fig. \ref{psetup1} is in fact \emph{more} general than that
of Fig. \ref{psetup5}.
Interestingly, in the other direction, it does not appear that the
results for the seemingly more restrictive setting of Fig. \ref{psetup1} can be
deduced from those for the setting of Fig. \ref{psetup5}.
\end{note}

\section{Special Cases}
\subsection{Feedback Logic At Encoder}\label{encoderaction}
The basic framework in this subsection is the setting in Fig. \ref{psetup2}.
\begin{figure}[htbp] 
\begin{center}
\scalebox{0.65}{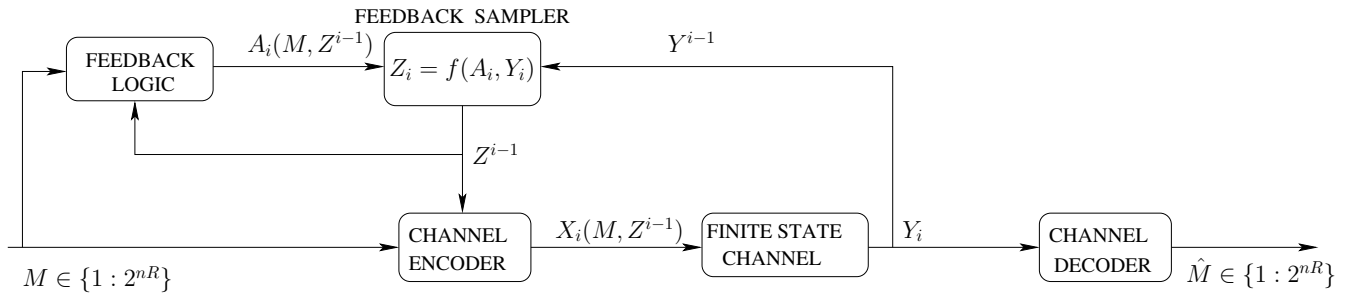}
\caption{Modeling \textbf{Feedback Sampling} for the acquisition of feedback in
Finite State Channels (FSCs) with only \textit{Encoder Feedback Logic}.}. 
\label{psetup2}
\end{center}
\end{figure}
Here only the encoder takes actions to govern feedback sampling.
\begin{theorem}
\label{theorem9}
 For no ISI, stationary and indecomposable FSC with encoder \textit{Feedback
Logic} as in Fig. 
\ref{psetup2}, the capacity is given by,
\bea
C_{enc}(\Gamma)=\lim_{N\rightarrow\infty}\frac{1}{N}\max_{Q(x^N,a^N\parallel
z^{N-1}),\E[\Lambda(A^N)]\le\Gamma}I(X^N\rightarrow Y^N).
\eea
\end{theorem}
\begin{proof}
 Specialize Theorem \ref{theorem8} as $A_d=\phi$, $A_e=A$ and
$\Lambda(A_e,A_d)=\Lambda(A_e)=\Lambda(A)$.
\end{proof}

\subsection{Feedback Logic At Decoder}\label{decoderaction}
In the previous sub-section, we characterized the fundamental limit for the
communication system 
as depicted in Fig. \ref{psetup2} where encoder took actions to govern
acquisition of feedback 
from decoder. However in some practical systems, it is the receiver (or decoder)
which estimates
the channel state perfectly and then decides to send it to the transmitter (or
encoder) through
noise-free feedback, \cite{GoldsmithVaraiya}. To model such a system where
sending noise free feedback from receiver
to transmitter is costly and is governed by actions taken by the decoder, we
consider the communication abstraction as in Fig. \ref{psetup3}.
\begin{theorem}
\label{theorem10}
 For no ISI, stationary and indecomposable FSC with decoder \textit{Feedback
Logic} as in Fig. 
\ref{psetup3}, the capacity is given by,
\bea
C_{dec}(\Gamma)=\lim_{N\rightarrow\infty}\frac{1}{N}\max_{Q(x^N\parallel
z^{N-1})Q(a^N\parallel
y^{N-1}),\E[\Lambda(A^N)]\le\Gamma}I(X^N\rightarrow Y^N).
\eea
\end{theorem}
\begin{figure}[htbp]
\begin{center}
\scalebox{0.65}{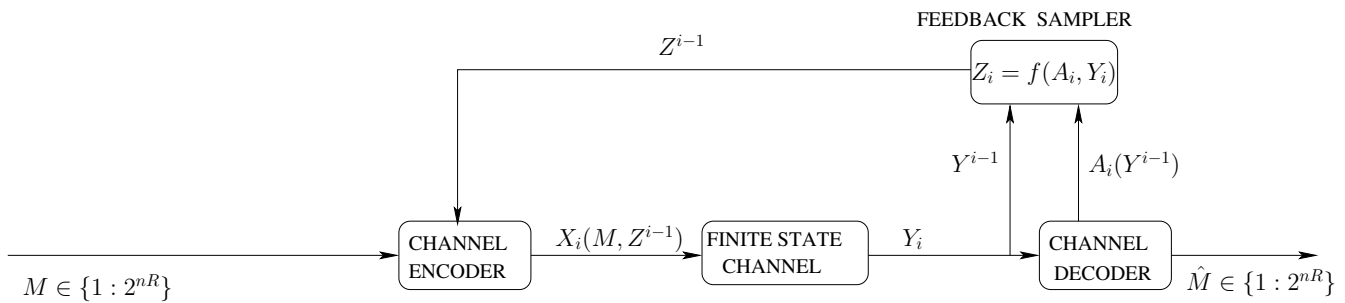}
\caption{Modeling \textbf{Feedback Sampling} for the acquisition of feedback in
Finite State Channels (FSCs) with only \textit{Decoder Feedback Logic}.}
\label{psetup3}
\end{center}
\end{figure}
\begin{proof}
 Specialize Theorem \ref{theorem8} as $A_e=\phi$, $A_d=A$ and
$\Lambda(A_e,A_d)=\Lambda(A_d)=\Lambda(A)$. 
\end{proof}
\begin{note}\label{note2}
Also note that if in Fig.
\ref{psetup3}, decoder can use current channel output to generate actions, we
can do the appropriate transformation in Theorem \ref{theorem12} to arrive at,
\bea
C_{dec,causal}(\Gamma)=\lim_{N\rightarrow\infty}\frac{1}{N}\max_{Q(x^N\parallel
z^{N-1})Q(\phi_d^N\parallel
y^{N-1}),\E[\Lambda(\Phi_d^N|_{Y^N})]\le\Gamma}I(X^N\rightarrow Y^N).
\eea
\end{note}

\section{Numerical Example 1 : To Feed or Not to Feed back}
\label{example1}
\begin{figure}[htbp]
\begin{center}
\scalebox{0.65}{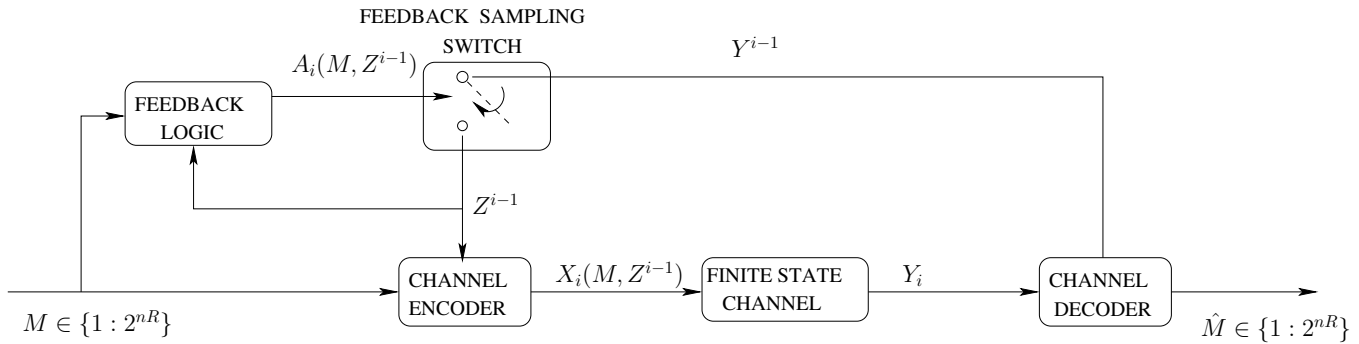}
\caption{Modeling logic \textquoteleft\textbf{to feed or not to
 feed back}\textquoteright\ in
Finite State Channels (FSCs) with encoder taking actions.}
\label{psetup4}
\end{center}
\end{figure}
\subsection{Encoder Actions}
Consider the setting as depicted in Fig. \ref{psetup4}. Now the actions are
binary, i.e., 
$\mathcal{A}=\{0,1\}$. In this setting, action sequence determine \textit{to
feed or not to feed back} a deterministic function of the past channel output,
i.e.,
\bea
Z_i&=&f(A_i,Y_i)=g(Y_i), \mbox{ if }A_i=1\nonumber\\
Z_i&=&f(A_i,Y_i)=\ast, \mbox{ }\mbox{ if }A_i=0,
\eea
where $\ast$ stands for erasure or no information about feedback. As a specific
example for such a setting
 consider the communication system involving Markovian channel as in Fig.
\ref{markovian}, which is essentially
a
no ISI, stationary, 
indecomposable FSC. Let the stationary distribution be given by $\pi(s)$,
$\forall s\in \mathcal{S}$. 
The feedback from the decoder at 
time $i$ consists of tuple $Y_{FB,i}=(Y_i,S_{i})$ and observed or
\textit{sampled} feedback 
$Z_i=g(Y_{FB,i})=S_{i}$, when $A_i=1$. The cost function, $\Lambda(a)=a$,
$a\in\mathcal{A}$
and the cost constraint is
 $\Gamma\in[0,1]$. Using Theorem \ref{theorem9}, the capacity of such a system
is given by,
\bea
C_{enc}(\Gamma)=\lim_{N\rightarrow\infty}C^N_{enc}(\Gamma)=\lim_{
N\rightarrow\infty } \frac { 1 } { N }
\max_{Q(x^N,a^N\parallel
z^{N-1}),\E[\Lambda(A^N)]\le\Gamma}I(X^N\rightarrow (
Y^N,S^N)).\label{noisicap}
\eea
In the following, we give single letter lower bound on $C_{enc}(\Gamma)$.
\begin{theorem}
\label{theorem11}
The capacity of the system in Fig. \ref{markovian} with encoder feedback logic
is
lower bounded as,
\bea
 C_{enc}(\Gamma)\ge C_{enc,lower}(\Gamma)=\max I(X;Y|S),
\eea
where maximization is over joint probability distribution, 
\bea
P_{S,A,Z,X,Y}(s,a,z,x,y)=\pi_S(s)
P_A(a)\1_{\{z=f(a,s)\}}P_{X|Z,A}(x|z,a)P_{Y|X,S}(y|x,s).
\eea
and $\E[\Lambda(A)]\le\Gamma$.
\end{theorem}
\begin{figure}[htbp]
\begin{center}
\scalebox{0.65}{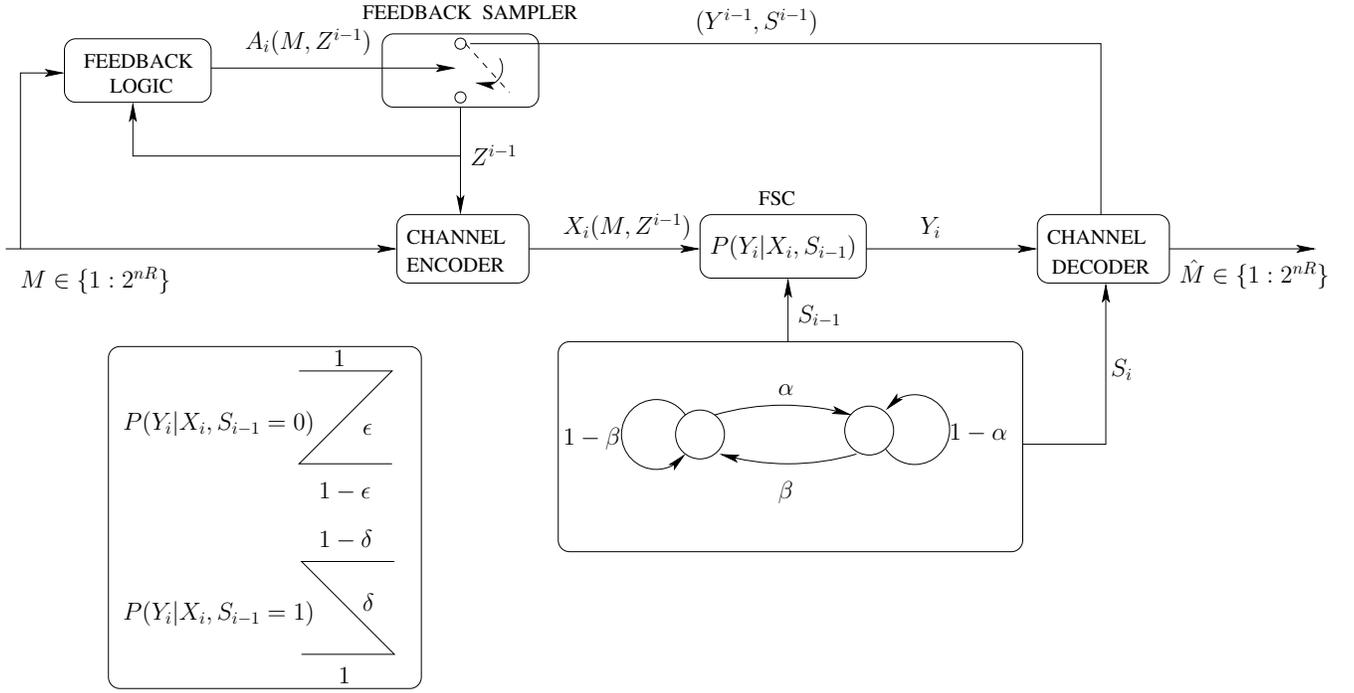}
\caption{\textbf{To feed or not to feed back} when encoder takes
actions and decoder knows the state. States are stationary and evolve as a
markov process.}
\label{markovian}
\end{center}
\end{figure}
\begin{proof}
The joint distribution in maximization in Eq. (\ref{noisicap}) is
\bea
P(s_0^N,a^N,z^N,x^N,y^N)=P(s_0)\prod_{i=1}^N 
Q(x_i,a_i|x^{i-1},a^{i-1},z^{i-1})P(y_i|x_i,s_{i-1})P(s_i|s_{i-1})\1_{\{
z_i=f(a_i ,
s_i)\}}.
\eea
To derive the lower bound we consider the following special type of above
distribution,
\bea
P'(s_0^N,a^N,z^N,x^N,y^N)=P(s_0)\prod_{i=1}^N Q(a_i)
Q(x_i|z_{i-1})P(y_i|x_i,s_{i-1})P(s_i|s_{i-1})\1_{\{z_i=f(a_i,s_i)\}}.
\eea
Note that right hand side of the above distribution can be factorized as,
\bea
P'(s_0^N,a^N,z^N,x^N,y^N)=\Phi_i(a_i,a_{i-1},s_{i-1},z_{i-1},s_i,z_i,x_i,y_i)
\Phi_{n\backslash i}(a^{i-2},a_{i+1}^N,s_0^{n\backslash
i},z^{i-2},z_{i+1}^N,x^{n\backslash i},y^{n\backslash i}),
\eea
which proves the markov chain,
\bea
(Y_i,S_i)-S_{i-1}-(Y^{i-1},S_0^{i-2}).\label{MC0}
\eea
Hence we have,
\bea
C^N_{enc}&\ge&\max\frac{1}{N}I(X^N\rightarrow (Y^N,S^N))\\
&=&\max\frac{1}{N}\sum_{i=1}^NI(X^i;Y_i,S_i|Y^{i-1},S^{i-1})\\
&=&\max\frac{1}{N}\sum_{i=1}^N H(Y_i,S_i|S_{i-1},S^{i-2},Y^{i-1})
-H(Y_i,S_i|X_i,S_{i-1},X_{i-1},S^{i-2},Y^{i-1})\\
&\stackrel{(a)}{=}&\max\frac{1}{N}\sum_{i=1}^N H(Y_i,S_i|S_{i-1})
-H(Y_i,S_i|X_i,S_{i-1})\\
&\stackrel{}{=}&\max\frac{1}{N}
\sum_{i=1}^N I(X_i;Y_i,S_i|S_{i-1})\\
&\stackrel{(b)}{=}&\max\frac{1}{N}\sum_{i=1}^N
I(X_i;Y_i|S_{i-1}),\label{num1}\\
\eea
where 

\begin{itemize}
 \item (a) follows from Markov Chain (\ref{MC0}) and from the channel model 
assumption [Eq. \ref{channelmodel}].
 \item (b) follows from the chain rule 
\bea
I(X_i;Y_i,S_i|S_{i-1})=I(X_i;Y_i|S_{i-1})+I(X_i,S_i|S_{i-1},Y_i),
\eea
and from the fact that $I(X_i,S_i|S_{i-1},Y_i)=0$, due the following state
evolution,
\bea
P(S_i|S_{i-1})=P(S_i|S_{i-1},X_i,Y_i),
\eea
which follows from the assumption on FSC for example in Fig. \ref{markovian} as,
\bea
P(Y_i,S_i|X_i,S_{i-1})=P(Y_i|X_i,S_{i-1})P(S_i|S_{i-1}).
\eea
\end{itemize}
The maximum in the above inequalities is taken over set of the distributions, 
\bea
\mathcal{S}_1=\{\prod_{i=1}^NQ(a_i)Q(x_i|z_{i-1}):\E[\Lambda(A^N)]\le\Gamma\}.
\eea
Clearly
\bea
\mathcal{S}_2=\{\prod_{i=1}^NQ(a_i)Q(x_i|z_{i-1}),\E[\Lambda(A_i)]\le\Gamma,
\mbox{ }\forall i\}\subseteq\mathcal{S}_1.
\eea
Now since the channel is stationary, it is invariant in time shift, hence
$P(s_i)=\pi(s_i)$, $\forall i$. 
Therefore we have the lower bound, 
\bea
C_{enc,lower}(\Gamma)&=&\lim_{N\rightarrow\infty}\frac{1}{N}
\max_{\prod_{i=1}^NQ(a_i)Q(x_i|z_{i-1}):\E[\Lambda(A_i)]\le\Gamma}\sum_{i=1}^N
I(X_i;Y_i|S_{i-1})\label{cap0}\\
&\stackrel{(c)}{\le}&\lim_{N\rightarrow\infty}\frac{1}{N}
\sum_{i=1}^N \max_{\prod_{i=1}^NQ(a_i)Q(x_i|z_{i-1}):\E[\Lambda(A_i)]\le\Gamma}
I(X_i;Y_i|S_{i-1})\\
&\stackrel{(d)}{=}&\lim_{N\rightarrow\infty}\frac{1}{N}
\sum_{i=1}^N \max_{P(a_{i-1},s_{i-1},z_{i-1},x_i,y_i)}I(X_i,Y_i|S_{i-1}),
\eea
where
\begin{itemize}
 \item (c) follows from the identity, $\max_a [f(x)+g(x)]\le \max_a f(x)+\max_a
g(x)$.
 \item (d) has,
\bea
P(a_{i-1},s_{i-1},z_{i-1},x_i,y_i)=\pi(s_{i-1})Q(a_{i-1})\1_ {\{ z_ { i-1 }
=f(a_{i-1}
,
s_{i-1})\}}Q(x_i|z_{i-1}
)P(y_i|x_i,s_{i-1}),
\eea
such that $\E[\Lambda(A_{i-1})]\le\Gamma$. $a_0$ is assumed to be a constant.
\end{itemize}
The inequality (c) holds with equality iff,
\bea
P(a_{i-1},s_{i-1},z_{i-1},x_i,y_i)=\argmax_{P(a,s,z,x,y).\E[\Lambda(A)]\le\Gamma
}I(X;Y|S)\mbox{ }\forall i,
\eea
where 
\bea
P(a,s,z,x,y)=\pi(s)Q(a)\1_{\{z=f(a,s)\}}Q(x|z)P(y|x,s).
\eea
Note that for our setting, $P(x|z,a)=P(x|z)$ as knowing $z$ determines
$a$. Therefore, we have
\bea
C_{enc,lower}(\Gamma)=\max I(X;Y|S).
\eea
with maximization over the joint distribution,
\bea
P_{S,A,Z,X,Y}(s,a,z,x,y)=\pi_S(s)
P_A(a)\1_{\{z=f(a,s)\}}P_{X|Z,A}(x|z,a)P_{Y|X,S}(y|x,s),
\eea
such that  $\E[\Lambda(A)]\le\Gamma$.
\end{proof}

\begin{note}
Note that lower bound on capacity at zero cost is,
\bea
C_{enc,lower}(\Gamma=0)=\max_{\pi_SP_XP_{Y|X,S}}I(X;Y|S).
\eea
This is indeed also the capacity at zero cost as derived below,
Clearly 
\bea \label{lowerbound}C_{enc}(\Gamma=0)\ge
C_{enc,lower}(\Gamma=0).
\eea 
Now
\bea
C^N_{enc}(\Gamma=0)&\stackrel{(a)}{=}&\frac{1}{N}\max_{Q(x^N)} I(X^N;Y^N,S^N)\\
&\stackrel{(b)}{=}&\frac{1}{N}\max_{Q(x^N)} I(X^N;Y^N|S^N)\\
&=&\frac{1}{N}\max \sum_{i=1}^N H(Y_i|S_{i-1},S^{n\backslash
(i-1)},Y^{i-1})-H(Y_i|X_i,S_{i-1},S^{n\backslash(i-1)},
X^{n\backslash i},Y^{i-1})\\
& \stackrel{(c)}{\le}&\frac{1}{N} \max \sum_{i=1}^N
H(Y_i|S_{i-1})-H(Y_i|X_i,S_{i-1})\\
&\stackrel{(d)}{=}&\frac{1}{N} \sum_{i=1}^N \max I(X_i;Y_i|S_{i-1})\\
&\stackrel{(e)}{=}&\max_{\pi_SP_XP_{Y|X,S}}I(X;Y|S)\label{upperbound},
\eea
where
\begin{itemize}
 \item (a) follows from the fact that $A_i=0$ for all $1\le i\le N$ and hence
since there is no feedback mutual
 information is equal to directed information (\cite{Massey}).
 \item (b) follows from the fact that $X^N$ and $S^N$ are independent.
  \item (c) follows from the fact that conditioning reduces entropy and from
the channel model 
assumption for the FSC,
i.e. Eq. (\ref{channelmodel}).
\item (d) follows from the identity, $\max_a [f(x)+g(x)]\le \max_a f(x)+\max_a
g(x)$.
\item (e) follows from the fact that maximization in (e) is on the joint,
\bea
P(x^N,s_0^N,y^N)=P(x^N)P(s_0)\prod_{i=1}^{N}P(y_i,s_i|x_i,s_{i-1}).
\eea
Hence the joint on $(X_i,Y_i,S_{i-1})$ is equivalent to,
\bea
P(X_i,Y_i,S_{i-1})=P(S_{i-1})P(X_i)P(Y_i|X_i,S_{i-1})=\pi(S_{i-1}
)P(X_i)P(Y_i|X_i,S_{i-1}).
\eea
\end{itemize}
Hence combining Eq. (\ref{lowerbound}) and (\ref{upperbound}) we establish
equality,
\end{note}

\begin{note}
 Just like capacity at zero cost, we can also show that lower bound on capacity
at unit cost is indeed tight
 too with similar steps as above. Hence we have,
\bea\label{cap00}
C_{enc}(\Gamma=1)=\max_{\pi_SP_{X|S}P_{Y|X,S}}I(X;Y|S).
\eea
Note that this scenario of complete feedback from decoder with state information
is similar to the 
scenario where encoder and decoder know the states. The capacity result for such
a communication system 
was characterized in \cite{CaireShamai} for channels with memory and indeed it
coincides with Eq. (\ref{cap00}).
\end{note}

\begin{note}
 It is interesting to observe that the lower bound on capacity is the
\textit{Probing Capacity} 
of the system considered 
in \cite{HimanshuHaimTsachyProbing}(as depicted in Fig. \ref{probing}), where
\begin{itemize}
 \item Channel is memoryless with state distribution that is i.i.d. as the
stationary
distribution $\pi_S$ of the channel with
memory considered here.
 \item Encoder takes message dependent actions that are binary and decide 
\textbf{to 
observe or not to observe} channel state.
 \item Decoder has complete state information.
\end{itemize}
\begin{figure}[htbp]
\begin{center}
\scalebox{0.7}{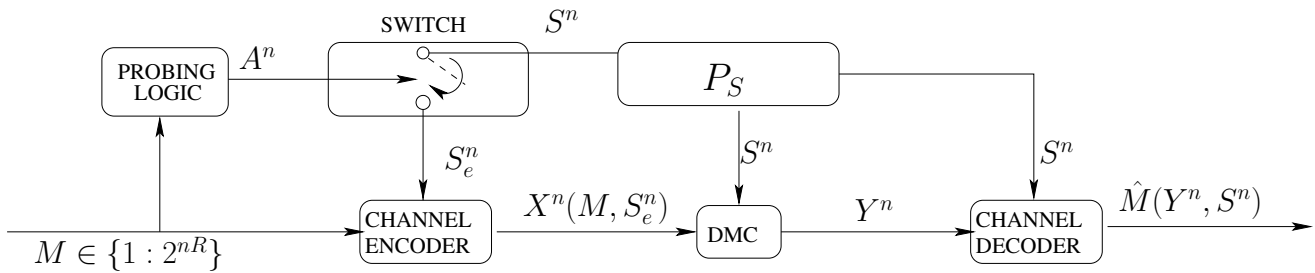}
\caption{The \textbf{Probing Capacity} setting in
\cite{HimanshuHaimTsachyProbing} where encoder takes
message dependent actions to observe state, encodes
using partial state information non-causally
while decoder knows the complete channel state. Note that here the state is
i.i.d.}
\label{probing}
\end{center}
\end{figure}
\end{note}

\subsection{Decoder Actions}
\begin{theorem}
\label{theorem13}
 Consider again the system in Fig. \ref{markovian} but now with decoder feedback
logic (instead of encoder feedback logic) with decoder taking actions causally
dependent on the channel output and state (for the sake of simplicity of
notation here we denote this capcity by $C_{dec}(\Gamma)$ instead of
$C_{dec,causal}(\Gamma)$). The capacity
of
such 
a system is lower bounded as,
\bea
 C_{dec}(\Gamma)\ge C_{dec,lower}(\Gamma)=\max I(X;Y|S),
\eea
where maximization is over joint probability distribution, 
\bea
P_{S,A,Z,X,Y}(s,a,z,x,y)=\pi_S(s)
P_{A|S}(a|s)\1_{\{z=f(a,s)\}}P_{X|Z,A}(x|z,a)P_{Y|X,S}(y|x,s),
\eea
such that $\E[\Lambda(A)]\le\Gamma$.
\end{theorem}
\begin{proof}
\textit{Note} the only difference in this lower bound as compared to
$C_{enc,lower}$ is that in that 
there is $P_{A|S}$ instead of $P_A$ in the distribution of maximization. The
proof is similar 
to proof of Theorem \ref{theorem11}, except the set of maximizing
distributions taken here is,
\bea
\mathcal{S}=\{\prod_{i=1}^{N}Q(a_i|s_{i})Q(x_{i}|z_{i-1}), \E[\Lambda(A_i)]\le
\Gamma\}.
\eea
Therefore,
\bea
C_{dec}(\Gamma)\ge
C_{dec,lower}(\Gamma)=\lim_{N\rightarrow\infty}\frac{1}{N}\max_{\mathcal{S}}
I(X^N\rightarrow (Y^N,S^N)).
\eea
All the other steps follow as in proof of Theorem \ref{theorem11}.
\end{proof}

We evaluate the lower bounds for the example in Fig. \ref{markovian},
when $\alpha=\beta=\epsilon=\delta=0.5$.
The region is shown in Fig. \ref{plot}. From the plot it is clear we can do
much better than time sharing between capacity at zero and unit cost when
either encoder or decoder takes actions.

\begin{figure}[htbp]
\begin{center}
\includegraphics[scale=0.35]{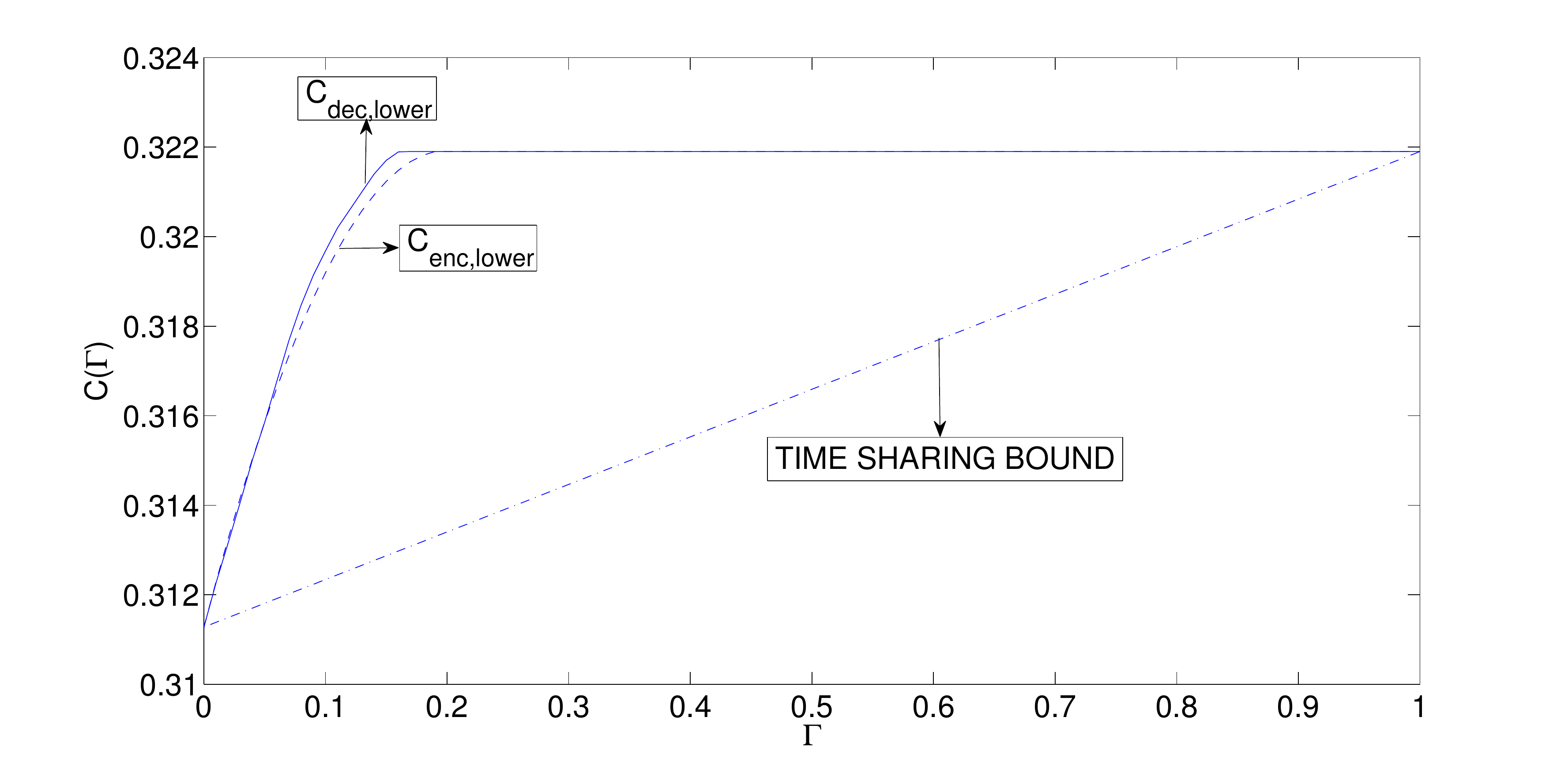}
\caption{Cost-capacity trade off for example in Fig. \ref{markovian}.
$C_{enc,lower}$ is the lower bound on capacity with encoder feedback
logic. If instead of encoder decoder decides (causally dependent on channel
output and state) when encoder will sample
feedback, $C_{dec,lower}$ is a lower bound on the capacity. The straight
time represents time sharing scheme which is strictly
sub-optimal.}
\label{plot}
\end{center}
\end{figure}

\section{Numerical Example 2 : Coding on the Backward Link in FSC}
\label{example2}
Consider the setting depicted in Fig. \ref{activefeedback}. We allow
\textit{coding on the backward link}, i.e., decoder encodes the channel
outputs causally ($A_i(Y^{i})\in\mathcal{A}$) and sends it to
the encoder. The encoder uses the acquired \textit{active feedback} symbols to
generate channel input symbols, i.e., $X_i(M,A^{i-1})$. For
stationary indecomposable FSCs with active feedback we denote the
capacity by $C_{AF}$.
The setting in Fig. \ref{activefeedback} is a very special case of
the framework of decoder feedback logic considered in Note 
\ref{note2} at the end of Section \ref{decoderaction}, when we let,
$f(A_{d,i},Y_i)=A_{d,i}=A_i$. Hence using the above conditions in
the capacity expression mentioned in Note \ref{note2} at the end of Theorem
\ref{theorem13} we have,
\bea
C_{AF}=\lim_{N\rightarrow\infty}\frac{1}{N}\max I(X^N\rightarrow Y^N),
\eea
where maximization is over joint distribution,
\bea
P(x^N,a^N,\phi_d^N,y^N)=Q(x^N\parallel
a^{N-1})Q(\phi_d^N\parallel y^{N-1})P(y^N\parallel
x^N)\prod_{i=1}^{N}\1_{\{a_{i}=\phi_{d,i}|_{y_i}\}},
\eea
such that $\E[\Lambda(a^N)]=\E[\Lambda(\Phi_d^N|_{Y^N})]\le\Gamma$, where
$\Phi_d^N|Y^N,\phi_d^N$ are as defined in Section \ref{causaldecoding}.
\par
Now consider an example under this setting where the channel evolution is
markovian with
binary states, i.e. 
\bea
P(Y_i,S_i|X_i,S_{i-1})=P(Y_i|X_i,S_{i-1})P(S_i|S_{i-1}),
\eea 
and states are known to the decoder on the fly. Hence,
decoder performs coding on backward link as,
$A_i=A_i(Y^{i},S^{i})\in\mathcal{A}$. The markov chain is assumed to be
stationary with distribution $\pi_S$ and states take values in a finite alphabet
$\mathcal{S}$. We consider following special cases when
$\card{\mathcal{A}}\ge\card{\mathcal{S}}$ : 
\begin{figure}[htbp]
\begin{center}
\scalebox{0.65}{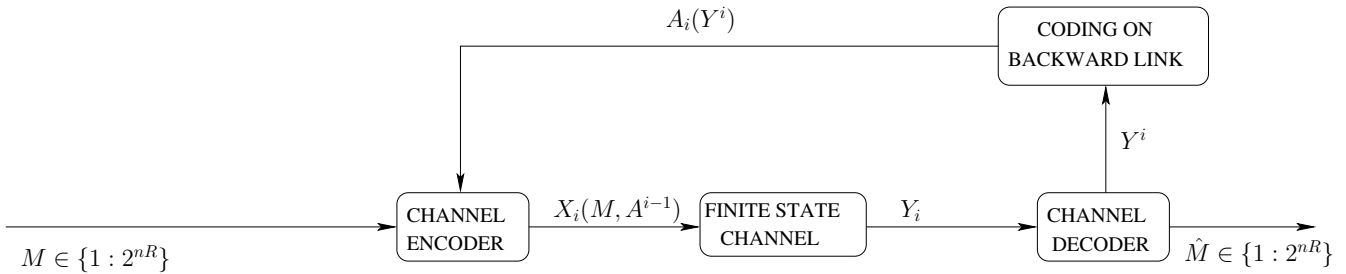}
\caption{Modeling \textbf{coding on backward link} in finite state channels
(FSCs).}
\label{activefeedback}
\end{center}
\end{figure}
\subsection{No Cost Constraints}
\begin{theorem}
 \label{theorem15}
Under this setting, the capacity is given by,
\bea
C_{AF}=\max_{\pi_SP_{X|S}P_{Y|X,S}}I(X;Y|S).
\eea
\end{theorem}
\begin{proof}
Since the decoder knows the states, the effective output is the tuple
$Y_{FB,i}=(Y_i,S_i)$.
Achievability is
straightforward. Actions basically communicate the state,
$A_i=S_{i}$. Hence in this case, 
the setup is same
as
encoder and decoder knowing states, and by the notes at the end of Section
\ref{encoderaction}, we have,
\bea\label{eq166}
C_{AF}\ge\max_{\pi_SP_{X|S}P_{Y|X,S}}I(X;Y|S).
\eea
Now consider for the converse,
\bea
C_{AF}&=&\lim_{N\rightarrow\infty}\max \frac{1}{N}I(X^N\rightarrow (Y^N,S^N))\\
&=&\lim_{N\rightarrow\infty}\max \frac{1}{N} \sum_{i=1}^N
H(Y_i,S_i|Y^{i-1},S^{i-1})-H(Y_i,S_i|X^i,Y^{i-1},S^{i-1})\\
&\stackrel{}{\le}&\lim_{N\rightarrow\infty}\max \frac{1}{N} \sum_{i=1}^N
H(Y_i,S_i|S_{i-1})-H(Y_i,S_i|X_i,S_{i-1})\\
&\stackrel{}{=}&\lim_{N\rightarrow\infty}\max \frac{1}{N} \sum_{i=1}^N
I(X_i;Y_i,S_i|S_{i-1})\\
&\stackrel{(a)}{=}&\lim_{N\rightarrow\infty}\max \frac{1}{N} \sum_{i=1}^N
I(X_i;Y_i|S_{i-1})\\
&\stackrel{(b)}{\le}&\lim_{N\rightarrow\infty}  \frac{1}{N}\sum_{i=1}^N
\max I(X_i;Y_i|S_{i-1}),
\eea
where (a) follows from the similar arguments used for Eq. (\ref{num1}) while
(b) follows from the identity $\max_a [f(x)+g(x)]\le \max_a f(x)+\max_a
g(x)$. Note that maximization in (b) is over the joint probability distribution,
\bea
P(s_0,x^N,a^N,\phi_d^N,y^N)=Q(x^N\parallel
a^{N-1})Q(\phi_d^N\parallel y^{N-1})P(y^N\parallel
x^N,s_0)\prod_{i=1}^{N}\1_{\{a_{i}=\phi_{d,i}|_{y_i}\}},
\eea
but one can average out, $(a^N,\phi_d^N,s^{n\backslash {i-1}},x^{n\backslash
i},y^{n\backslash i})$ so that for the $i^{th}$ term, maximization is over the
joint probability distribution,
\bea
P(s_{i-1},x_i,y_i)=\pi_S(s_{i-1}) P(x_i|s_{i-1})P(y_i|x_i,s_{i-1}),
\eea
since states are stationary and distributed as $\pi_S$.
Hence we have,
\bea\label{eq175}
C_{AF}&\le&\max_{\pi_SP_{X|S}P_{Y|X,S}}I(X;Y|S).
\eea
Proof is completed using Eq. (\ref{eq166}) and (\ref{eq175}).
\end{proof}
\subsection{Cost constraint $\Gamma$} The condition of this subsection
differs from those of the previous in cost constraints. The intuition here is to
look for an achievability scheme which decides when to send or not send state
information from decoder to encoder depending on cost constraints.
\begin{figure}[htbp]
\begin{center}
\scalebox{0.65}{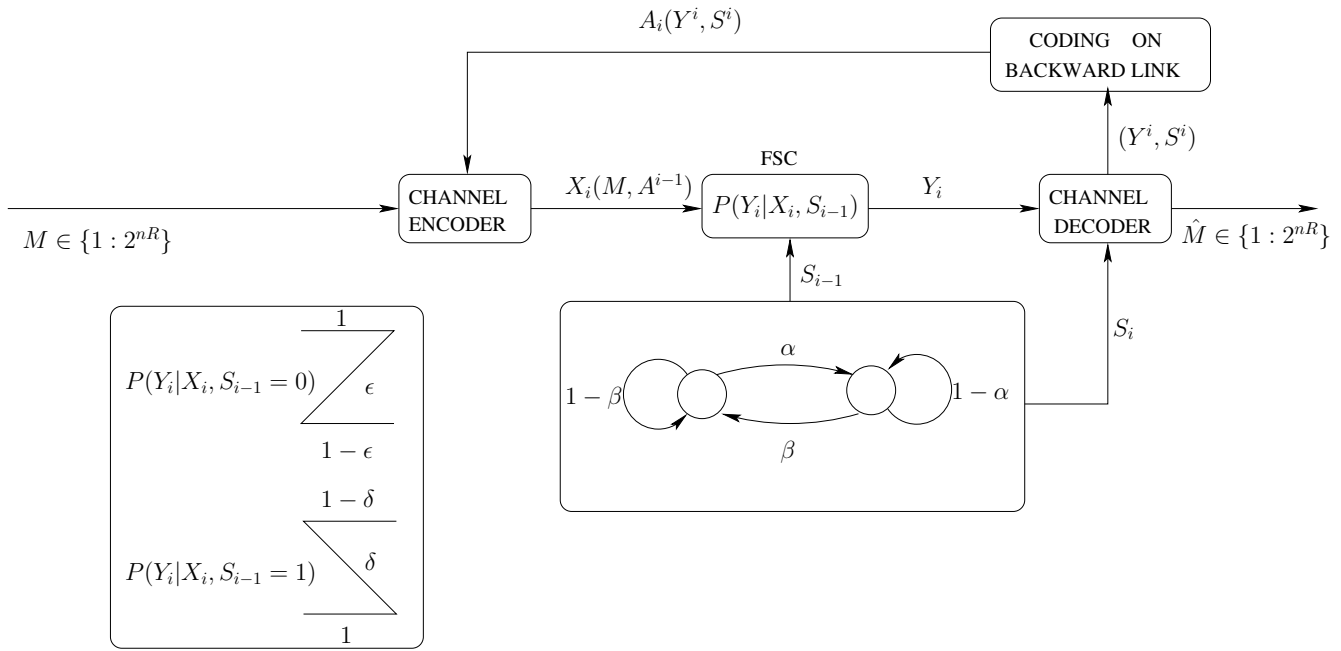}
\caption{Modeling \textbf{coding on the backward link} for the Markovian channel
with
binary states.}
\label{activefeedbackexample}
\end{center}
\end{figure}
\begin{theorem}
For the system in Fig. \ref{activefeedbackexample}, the capacity is lower
bounded as,
\bea
 C_{AF}(\Gamma)\ge C_{AF,lower}(\Gamma)=\max I(X;Y|S),
\eea
where maximization is over joint probability distribution, 
\bea
P_{S,A,X,Y}(s,a,x,y)=\pi_S(s)
P_{A|S}(a|s)P_{X|A}(x|a)P_{Y|X,S}(y|x,s),
\eea
where $\E[\Lambda(A)]\le\Gamma$.
\end{theorem}
\begin{proof}
We outline only the sketch of the proof as it is similar 
to the proof of Theorem \ref{theorem11}, except the set of maximizing
distributions taken here is,
\bea
\mathcal{S}=\{\prod_{i=1}^{N}Q(a_i|s_i)Q(x_{i}|a_{i-1}), \E[\Lambda(A_i)]\le
\Gamma\}.
\eea
Therefore,
\bea
C_{AF}(\Gamma)\ge
C_{AF,lower}(\Gamma)=\lim_{N\rightarrow\infty}\frac{1}{N}\max_{\mathcal{S}}
I(X^N\rightarrow (Y^N,S^N)).
\eea
All the other steps follow as in proof of Theorem \ref{theorem11}.
\end{proof}
 We consider
an example under this setting as depicted in Fig. \ref{activefeedbackexample}.
We assume $\mathcal{A}$ is a binary alphabet with cost function,
$\Lambda(a)=a$, $a\in\{0,1\}$, hence this models the scenario of
cost constrained \textit{one-bit active feedback} in the given finite state
channel. The plot for $\alpha=\beta=\delta=0.5$ is shown in Fig.
\ref{exampleaf}. Note that this bound is equal to the bound,
$C_{dec,lower}$ in Theorem \ref{theorem13} for $f(a,s)=a$.
\begin{figure}[htbp]
\begin{center}
\includegraphics[scale=0.32]{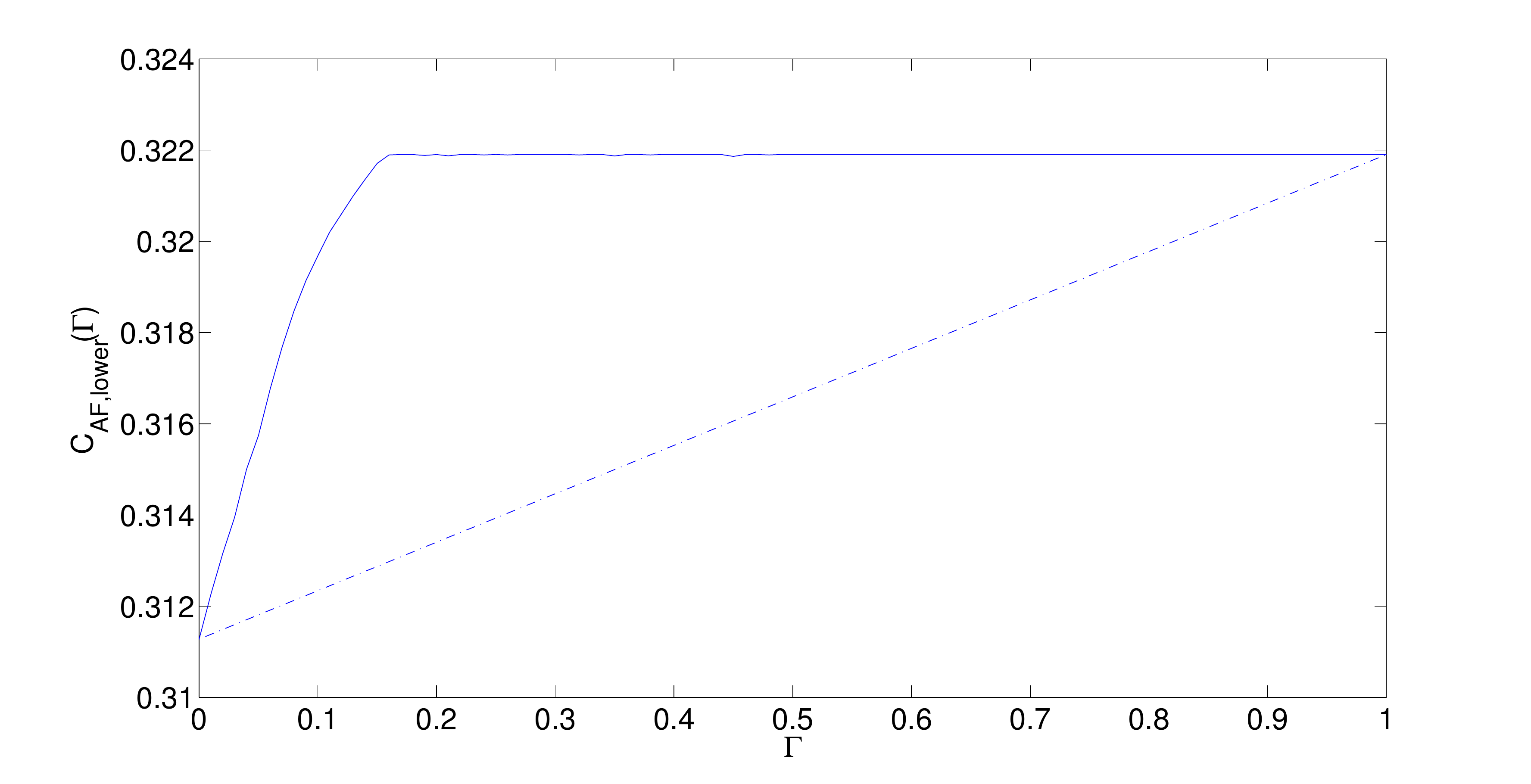}
\caption{Cost-capacity trade off for example in Fig.
\ref{activefeedbackexample}.
$C_{AF,lower}$ is the lower bound on capacity. The dashed straight line represents
a naive time sharing scheme. It is seen that not only is time-sharing suboptimal, but that the feedback capacity can be achieved in full even if observing only a small fraction of the symbols fed back. 
}
\label{exampleaf}
\end{center}
\end{figure}
\section{Blahut Arrimoto Algorithm for Action Dependent Feedback}
\label{algorithm}
In this section we will develop a numerical algorithm, as an extension to Blahut
Arrimoto Algorithm for Action dependent
feedback (BAA-Action), when the encoder takes the action to determine the
quality and availability of the feedback from the decoder. Blahut, \cite{Blahut} and Arimoto, \cite{Arimoto}
suggested an algorithm based on alternating maximization to compute the mutual
information, and this was extended to computing directed information in
\cite{Naiss_Permuter_BAA}. Our approach is similar to the latter, with the
difference being that in our case the aqcuisition of the feedback
is determined by a cost constrained action. Our algorithm also works for the case when 
there is a joint cost constraint on both action and channel input symbols, thereby generalizing the result in \cite{Naiss_Permuter_BAA} to compute directed information with cost constraints on channel input symbols. 
\subsection{Algorithm : BAA-Action} \label{subsec:BAA}
\begin{algorithm}
\caption{BAA-Action : Block length $N$ and the channel $p(y^N\parallel x^N)$ is given. The Lagrangian multiplier $\lambda$ is fixed. }
\begin{algorithmic}
\label{BAA_Action}
\State \textbf{Initialize} $q(x^N,a^N|y^N)$. For eg. initializing with uniform
distribution, i.e., $q(x^N,a^N|y^N)\gets\ \card{\mathcal{X}}^{-N}$.
\For {$i = N \to 1 $}
\State
$r'(x^i,a^i,z^{i-1})\gets \prod_{x_{i+1}^N,a_{i+1}^N,y_i^N}\prod_{\Ascr_{i,z}}
\Bigg { [ }
\frac{q(x^N,a^N|y^N)2^{-N\lambda\Lambda(a^N)}}{\prod_{j=i+1}^Nr(x_j,
a_j|x^{j-1},a^{j-1},z^{j-1})}\Bigg{]}^{\frac{p(y^N\parallel
x^N)\prod_{j=i+1}^{N}r(x_j,a_j\parallel
x^{j-1},a^{j-1},z^{j-1})}{\sum_{A_{i,z}}\prod_{j=1}^{i-1}p(y_j|x^j,y^{j-1})}}
$\newline\newline
\State 
$r(x_i,a_i|x^{i-1},a^{i-1},z^{i-1})\gets\frac{r'(x^i,a^i,z^{i-1})}{\sum_{x_i ,
a_i}
r'(x^i,a^i,z^{i-1})}$\newline
\State where $A_{i,z}=\{y^{i-1}:f(a^{i-1},y^{i-1})=z^{i-1}\}$
\EndFor
\State \textbf{Compute} $r(x^N,a^N\parallel
z^{N-1})\gets \prod_{i=1}^{N}r(x_i,a_i|x^{i-1},a^{i-1},z^{i-1})$
\newline
\State \textbf{Compute} $q(x^N,a^N|y^N)\gets \frac{r(x^N,a^N\parallel
z^{N-1})p(y^N\parallel
x^N)}{\sum_{x^N,a^N}r(x^N,a^N\parallel z^{N-1})p(y^N\parallel x^N)}$
\newline
\State \textbf{Calculate} $I_U-I_L$ where,\newline
\State $I_L\gets \frac{1}{N}\sum_{x^N,a^N,y^N}r(x^N,a^N\parallel
z^{N-1})p(y^N\parallel x^N)\log
\frac{q(x^N,a^N|y^N)}{r(x^N,a^N\parallel z^{N-1})}.$
\State $I_U\gets \frac{1}{N}\max_{x_1,a_1}\sum_{y_1}\max_{x_2,a_2}\cdots\max_{
x_N , a_N } \sum_ { y_N}p(y^N\parallel x^N)\log\frac{p(y^N\parallel
x^N)2^{-N\lambda\Lambda(a^N)}}{\sum_{x^N,a^N}p(y^N\parallel
x^N)r(x^N,a^N\parallel z^{N-1})}.$\newline
\If{$I_U-I_L>\epsilon$} Goto the loop again. 
\EndIf\newline
\State \textbf{Compute} $C_N^{(\lambda)}\gets I_U$.
\newline\State \textbf{Compute} $\Gamma^{(\lambda)}\gets \sum_{x^N,a^N,y^N}r(x^N,a^N\parallel
z^{N-1})p(y^N\parallel x^N)\Lambda(A^N)$.
\end{algorithmic}
\end{algorithm}
We formally state the algorithm which will
be used later to give a series of computable upper and lower bounds. The goal is to maximize the normalized directed information, 
\bea
\frac{1}{N}I(X^N\rightarrow Y^N)=\frac{1}{N}\sum_{x^N,a^N,y^N}p(x^N,a^N\parallel z^{N-1})p(y^n\parallel
x^N)\1_{\{z^N=f(a^N,y^N)\}} \log \frac{p(y^N\parallel x^N)}{p(y^N)}	
\eea
where maximum is over joint probability distributions $p(x^N,a^N\parallel
z^{N-1})p(y^n\parallel x^N)\1_{\{z^N=f(a^N,y^N)\}} $ such that
$\E[\Lambda(A^N)]\le \Gamma$. We will henceforth refrain from explicitly writing 
$\1_{\{z^N=f(a^N,y^N)\}} $ as it is clear from the context how $z^N$ is a
deterministic function of $x^N$ and $a^N$. Let us denote $p(x^N,a^N\parallel
z^{N-1})$ by $r(x^N,a^N\parallel z^{N-1})$. Note further that effectively the
maximization is over $r(\cdot)$ as $p(y^N\parallel x^N)$ is the property of the
channel and $r(\cdot)$ determines whether the cost constraints are satisfied or
not. Note that,
\bea
p(x^N,a^N,y^N)&=&r(x^N,a^N\parallel z^{N-1})p(y^N\parallel x^N)\\
&=&p(y^N)q(x^N,a^N|y^N).
\eea
Thus the directed information can equivalently be written as,
\bea
 I(X^N\rightarrow Y^N)&=&\sum_{x^N,a^N,y^N}r(x^N,a^N\parallel
z^{N-1})p(y^n\parallel x^N)\log \frac{q(x^N,a^N|y^N)}{r(x^N,a^N\parallel
z^{N-1})}\\
 &\stackrel{\triangle}{=}&\mathcal{I}(\mathbf{r},\mathbf{q})\label{cn},
\eea
where the shorthand notations are defined as
$\mathbf{r}\stackrel{\triangle}{=}r(x^N,a^N\parallel z^{N-1})$ and
$\mathbf{q}\stackrel{\triangle}{=}q(x^N,a^N | y^N)$.  
\par
Our algorithm (presented above as Algorithm 1) computes the normalized directed information (Eq. \ref{cn}) by using the Lagrangian approach (outlined in Appendix \ref{BAA-Algorithm}). The Lagrangian multiplier, $\lambda$, corresponds to the tradeoff between cost and the normalized directed information. Hence we define for
$\lambda\ge 0$,
$\mathcal{I}(\mathbf{r},\mathbf{q},\lambda)=\mathcal{I}(\mathbf{r},\mathbf{q}
)-\lambda\E[\Lambda(A^N)]$. Denote $C_N^{(\lambda)}=\max \mathcal{I}(\mathbf{r},\mathbf{q},\lambda)$ and $\Gamma^{(\lambda)}$ is the corresponding cost incurred at the maximizer of $C_N^{(\lambda)}$. The evaluations as described in Algorithm 1, of $C_N^{(\lambda)}$ and $\Gamma^{(\lambda)}$ together,  characterize this tradeoff, and this tradeoff curve is obtained by appropriately sweeping through the values of $\lambda$.
\par
The algorithm takes as input a particular block length $N$, channel $p(y^N\parallel x^N)$ and the Lagrangian multiplier $\lambda$. To begin with, $q(\cdot)$ is initialized with the uniform distribution as shown in Algorithm 1. This $q(\cdot)$ then is used to update $r(\cdot)$, which is then used to update $q(\cdot)$. The update rule is chosen so that directed information is maximized, hence this is an alternate maximization procedure, as in \cite{Blahut}, \cite{Arimoto}, \cite{Naiss_Permuter_BAA}. The lower and upper bounds $I_U, I_L$ converge to $C_N^{(\lambda)}$ for increasing number of iterations. 
\subsection{Numerical Evaluation}
Here we propose and evaluate upper and lower bounds for the Example described in Fig. \ref{markovian}, where states are generated through a Markovian channel and only the encoder takes actions to decide whether the states, which are known to the decoder, will be fed back to the encoder or not. We will also contrast the bounds with the analytical lower bound $C_{enc,lower}(\Gamma)$ obtained in Section \ref{example1}. Our investigation in this section yields a contrasting upper bound to the capacity, which along with analytical lower bound $C_{enc,lower}(\Gamma)$ provides tight and computable bounds for capacity.
\par
From Section \ref{example1}, the capacity of this channel (which is finite indecomposable) is 
\bea
\lim_{N\rightarrow\infty}\frac{1}{N}\max_{\mathbf{r}, s.t. \E[\Lambda(A^N)]\le\Gamma}I(X^N\rightarrow Y^N,S^N)&=&\lim_{N\rightarrow\infty}\frac{1}{N}\max_{\mathbf{r}, s.t. \E[\Lambda(A^N)]\le\Gamma}\min_{s\in\mathcal{S}}I(X^N\rightarrow Y^N,S^N|S_0=s)\\
&\stackrel{(a)}{=}&\lim_{N\rightarrow\infty}\frac{1}{N}\max_{\mathbf{r}, s.t. \E[\Lambda(A^N)]\le\Gamma}I(X^N\rightarrow Y^N,S^N|S_0=0)
\label{limit1}
\eea
where (a) follows from the structure of the channel due to the fact that Z and S channel induce similar joint distributions due to symmetry in their structure (one channel can be obtained by replacing 1 with 0 and 0 with 1), so we have $\max_{r(x^N,a^N\parallel z^{N-1},s_0=0)} I(X^N\rightarrow Y^N|S_0=0)=\max_{r(x^N,a^N\parallel z^{N-1},s_0=1)} I(X^N\rightarrow Y^N|S_0=1)$.  Define, 
\bea
C_N(\Gamma)&\stackrel{\triangle}{=}&\max \frac{1}{N} I(X^N\rightarrow Y^N,S^N|S_0=0)
\eea
where maximum is over joint probability distributions $r(x^N,a^N\parallel
z^{N-1})p(y^n\parallel x^N,s_0=0)\1_{\{z^N=f(a^N,y^N)\}} $ such that
$\E[\Lambda(A^N)]\le \Gamma$.  We now have the following theorem, (the proof is deferred to Appendix \ref{BAA-Algorithm}),
 \begin{theorem}\label{theorem-evaluation}
For the channel in Fig. \ref{markovian}, there exists computable bounds for $C(\Gamma)$ defined for $N\ge 1$, where the lower bound is, $C_N(\Gamma-\frac{1}{N})\le C(\Gamma)$ for $\Gamma \in[\frac{1}{N},1]$ and the upper bound is, $C(\Gamma)\le C_N(\Gamma)$, for $\Gamma\in[0,1]$.
 \end{theorem}
 \par
Note that the bounds are tight and converge to the capacity as $N\rightarrow\infty$. The computation is performed using the algorithm outlined above. For the case of $\alpha=\beta=\delta=\epsilon=0.5$, the bounds computed using the algorithm are shown in Fig. \ref{BAA-Action-bounds} where computation is performed for $N=2$ and $N=3$. Fig \ref{capacity-bounds} contrast the upper bound of $N=3$ along with the analytical lower bound $C_{enc,lower}(\Gamma)$. Note from the graph, the benefit of using the full feedback is obtained around $\Gamma\sim 0.2034$.  The code is available at \cite{BAA-Action-Code}.
\begin{figure}[htbp]
\begin{center}
\includegraphics[scale=0.5]{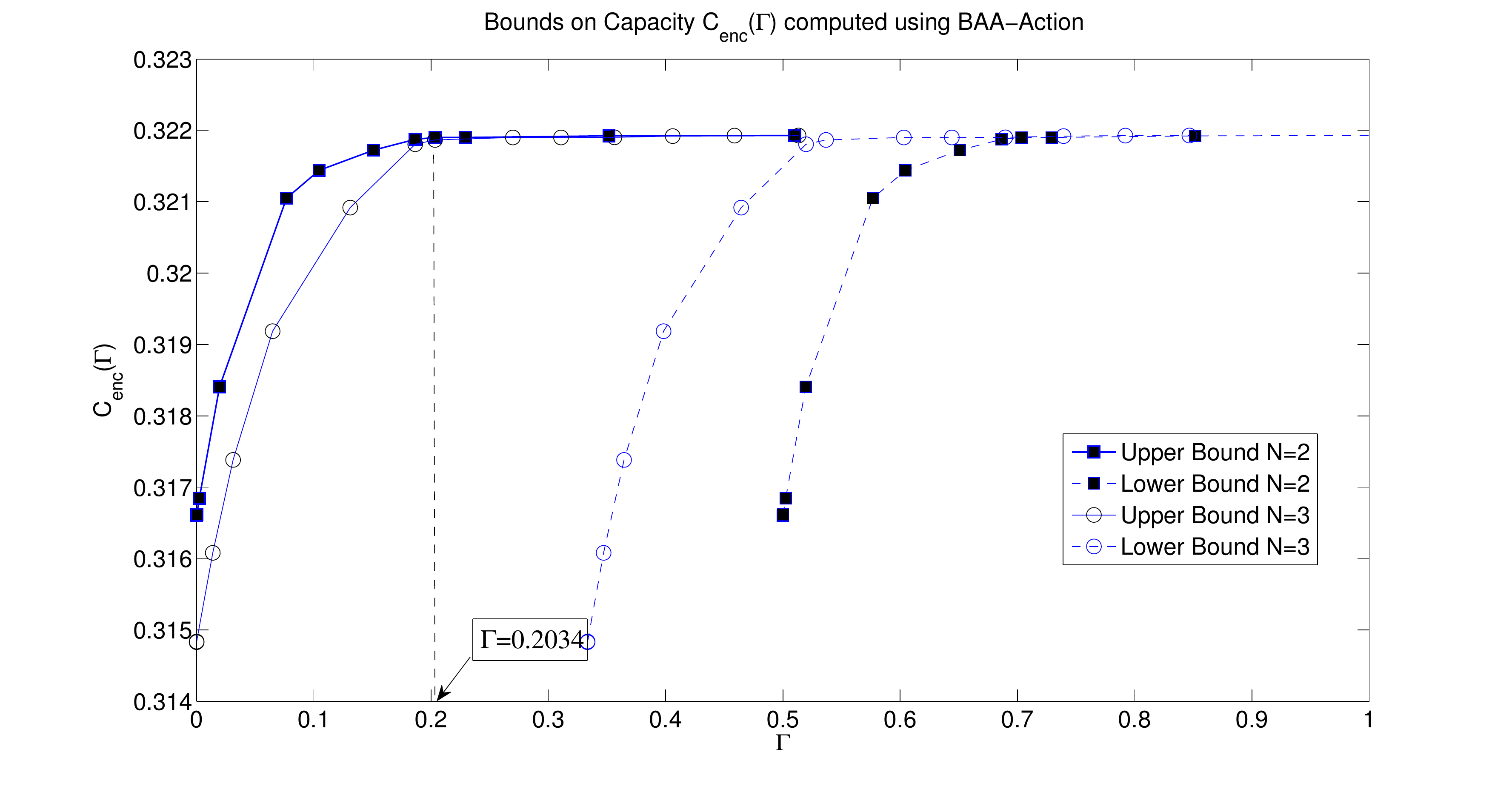}
\caption{Cost-capacity trade off for example in Fig.
\ref{markovian}. Bounds correspond to our calculation from the algorithm for $N=2$ and $N=3$.}
\label{BAA-Action-bounds}
\end{center}
\end{figure}
\begin{figure}[htbp]
\begin{center}
\includegraphics[scale=0.5]{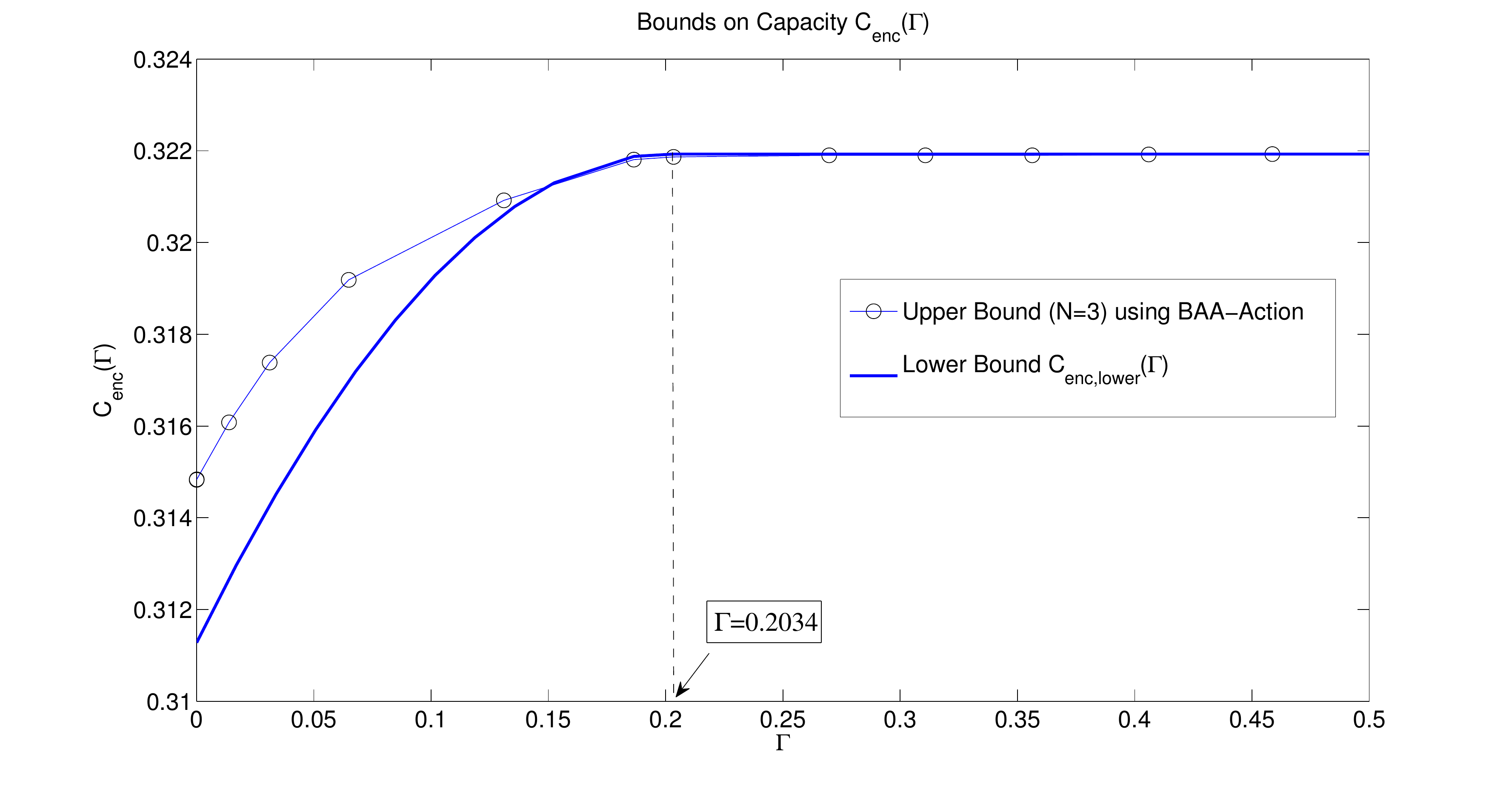}
\caption{Cost-capacity trade off for example in Fig.
\ref{markovian}. $C_{enc,lower}(\Gamma)$ is the lower bound on capacity computed analytically from the expression in Section \ref{example1}. The computed upper bound for $N=3$ using BAA-Action provides a tight upper bound to the capacity as shown in the figure. Note from the graph, the benefit of using the full feedback is obtained around $\Gamma\sim 0.2034$. }
\label{capacity-bounds}
\end{center}
\end{figure}

\section{Conclusion}
In this paper, we studied communication systems with finite state channels
(FSCs), where the encoder and decoder adaptively decide what to feed 
back from the decoder to encoder to optimize for the rate of reliable
communication, under an average cost constraint. For FSCs where
probability of initial state is
positive for
all states or for stationary indecomposable FSCs, we have the exact
characterization of the capacity. We also discuss the special case of \textit{to
feed or not to feed back} where either the encoder or the decoder takes binary
actions that determine whether or not a deterministic function of channel output
will be fed back to the encoder. As another special case, we characterize the
capacity in case of \textit{coding on the backward link} for FSCs. In case of
Markovian channels, with explicit computation we show that the naive time sharing schemes can be highly suboptimal. Finally, we proposed a Blahut-Arimoto type algorithm based on alternate maximization to give computable upper and lower bounds for a class of Markovian Channel.  

\label{conclusion}

\bibliography{fsc}
\bibliographystyle{IEEEtran}

\appendices
\section{Some Properties of Causal Conditioning and Directed Information}
\label{basic}
Here we present some of the basic properties of causal conditioning and directed
information. The proofs are omitted as being similar to the corresponding
Lemmas in \cite{HaimTsachyGoldsmith}.
\begin{itemize}
 \item \textit{Property 1 :}[Chain rule for causal conditioning]
\bea
P(x^N,a^N,y^N)=P(y^N\parallel x^N,a^N)P(x^N,a^N\parallel
y^{N-1}).
\eea
Similarly,
\bea
P(x^N,a^N,y^N,s_0)=P(y^N\parallel x^N,a^N,s_0)P(x^N,a^N\parallel
y^{N-1},s_0).
\eea
\item \textit{Property 2 :}
\par  
 $P(x^N,a_e^N\parallel z^{N-1})$ uniquely determines
$P(x_i,a_{e,i}|x^{i-1},a_e^{i-1},z^{i-1})$ $\forall$ $1\le i\le N$ and
all the arguments $(x^{i-1},a_e^{i-1},z^{i-1})$, for which
$P(x^{i-1},a_e^{i-1},z^{i-1})>0$. Similar results holds for $P(a_{d}^N\parallel
y^{N-1})$.
\item \textit{Property 3 :} 
\par
$|I(X^N\rightarrow Y^N)-I(X^N\rightarrow Y^N|S)|\le
H(S)\le \log \card{\mathcal{S}}$.
\end{itemize}


\section{Proof of Theorem \ref{theorem2}}\label{proof2}
From our achievability scheme and Lemma \ref{achjoint} we have
\bea
P(x^N,a_e^N,a_d^N,y^N)&=&\sum_{s_0}P(s_0,x^N,a_e^N,a_d^N,y^N)\\
&=&Q(x^N,a_e^N\parallel z^{N-1})Q(a_d^N\parallel
y^{N-1})\sum_{s_0}P(s_0)P(y^N\parallel x^N,s_0)\\
&=&Q(x^N,a_e^N\parallel z^{N-1})Q(a_d^N\parallel y^{N-1})P(y^N\parallel x^N),
\eea
where last equality follows from Eq. (\ref{channelmodel1}).
Hence we have
\bea
E(P_{e,m})&=&\sum_{y^N}\sum_{x^N,a_e^N,a_d^N}P(x^N,a_e^N,a_d^N,y^N)P(\mbox{error
} |m , x^N , a_e^N , a_d^N, 
y^N)\\
&=&\sum_{y^N}\sum_{x^N,a_e^N,a_d^N}Q(x^N,a_e^N\parallel
z^{N-1})Q(a_d^N\parallel
y^{N-1})P(y^N\parallel x^N)P(\mbox{error}|m,x^N,a_e^N,a_d^N,y^N).
\label{erroreqn}
\eea
Let $A_{m'}=\{\mbox{event such that } P_{y^N|m'}>P_{y^N|m}
\mbox{ for }m'\neq m \}$. Alternatively if for a message $m$,
encoder generated $x^N$, $A_{m'}=\{\mbox{event such
that } P(y^N\parallel x^{'N})>P(y^N\parallel x^N) \mbox{ for }
x^{'N}\neq x^N \}$
\bea
P(A_{m'}|m,x^N,a_e^N,a_d^N,y^N)]&=&\sum_{x^{'N},a_e^{'N},a_d^{'N}}Q(x^{'N},
a_e^{'N}\parallel
z^{N-1})Q(a_d^{'N}\parallel
y^{N-1})\1_{\{P(y^N\parallel x^{'N})>P(y^N\parallel x^N)\}}\\
P(A_{m'}|m,x^N,a_e^N,a_d^N,y^N)]&\le&\sum_{x^{'N},a_e^{'N},a_d^{'N}}Q(x^{'N},
a_e^{'N}\parallel
z^{N-1})Q(a_d^{'N}\parallel
y^{N-1})\left[\frac{P(y^N\parallel x^{'N})}{P(y^N\parallel
x^N)}\right]^s\mbox{
any $s>0$}.
\eea
Hence,
\bea
P(\mbox{error}|m,x^N,a_e^N,a_d^N,y^N)&=&P(\cup_{m'\neq
m}A_{m'}|m,x^N,a_e^N,a_d^N,y^N)\\
&\le&\min\left[\sum_{m'\neq m}P(A_{m'}|m,x^N,a_e^N,a_d^N,y^N),1\right]\\
&\le&\left[\sum_{m'\neq m}P(A_{m'}|m,x^N,a_e^N,a_d^N,y^N)\right]^\rho\mbox{,
for any $0\le\rho\le 1$} \\
&\le&\left[(M-1)\sum_{x^{'N},a_e^{'N},a_d^{'N}}Q(x^{'N},a_e^{'N}\parallel
z^{N-1})Q(a_d^{'N}\parallel
y^{N-1})\left(\frac{P(y^N\parallel
x^{'N})}{P(y^N\parallel x^N)}\right)^s\right]^\rho\label
{temp1}\nonumber.\\
\eea
Now substituting (\ref{temp1}) in (\ref{erroreqn}) and using
$s=\frac{1}{\rho+1}$ we obtain,
\bea
E_{P_{e,m}}&\le&(M-1)^\rho\sum_{y^N}\left[\sum_{x^N,a_e^N,a_d^N}Q(x^N,
a_e^N\parallel
z^{N-1})Q(a_d^N\parallel
y^{N-1})P(y^N\parallel
x^N)^{\frac{1}{\rho+1}}\right]^{\rho+1}.
\eea

\section{Proof of Some Markov Chains}\label{markov}
If a given scheme satisfies joint probability distribution , as in
Eq. (\ref{mainjoint}), we have the following markov chains,
\begin{itemize}
 \item [MC1] $Y_i-(X^i,Y^{i-1},S_0)-(M,A_e^i,A_d^i)$. 
 \item [MC2]
$(X_i,A_{e,i})-(X^{i-1},A_e^{i-1},Z^{i-1})-(Y^{i-1},A_d^{i-1},S_0)$.
 \item [MC3] $A_{d,i}-(Y^{i-1},A_d^{i-1})-(X^i,A_e^i,Z^{i-1},S_0)$.
\end{itemize}
To prove MC1, consider again the joint probability distribution induced by a
given scheme,
\bea\label{mainjoint1}
&&P_{M,A_{e}^N,A_d^{N},Z^N,X^N,S_0^N,Y^N,\hat{M}}(m,a_{e}^N,a_d^{N},z^N,x^N,
s_0^N ,
y^N , \hat{m})\nonumber\\
&=&\frac{1}{\card{\mathcal{M}}} P_S(s_0)\prod_{i=1}^{n}
\1_{\{a_{d,i}=f_{A_{d,i}}(y^{i-1)}\}}
\1_{\{a_{e,i}=f_{A_{e,i}}(m,z^{i-1})\}}\nonumber\\
&&\times\prod_{i=1}^{n}\1_{\{x_i=f_{e,i}(m,z^{i-1})\}}P(y_i,s_i|x_i,
s_{i-1})\1_{\{z_{i}=f(a_{e,i},a_{d,i},y_i)\}}\times\1_{\{\hat{m}=f_d(y^n)\}}
.\nonumber\\
\eea
Summing over
$(\hat{M},X_{i+1}^{N},Z_i^{N},S_{i+1}^N,Y_{i+1}^N,A_{e,i+1}^N,A_{d,i+1}^N)$ in
Eq. (\ref{mainjoint1}) we obtain,
\bea
&&P_{M,A_{e}^i,A_d^{i},Z^{i-1},X^i,S_0^i,Y^i}(m,a_{e}^i,a_d^{i},z^{i-1},x^i,
s_0^i,y^i)\\
&=&\frac{1}{\card{\mathcal{M}}}\prod_{j=1}^{i}
\1_{\{a_{d,j}=f_{A_{d,j}}(y^{j-1)}\}}
\1_{\{a_{e,j}=f_{A_{e,j}}(m,z^{j-1})\}}\1_{\{x_j=f_{e,j}(m,z^{j-1})\}}
\times\prod_{j=1}^{i-1}\1_{\{z_{i
}=f(a_{e,i},a_{d,i},y_i)\}}\nonumber\\&&\times\prod_{j=1}^{i}P(y_j,s_j|x_j,
s_{j-1})\times P_S(s_0)\label{mainjoint2}\\
&=&\Phi_1(M,A_e^i,A_d^i,Z^{i-1},Y^{i-1},X^i)\Phi_2(Y^i,S_0^i,X^i)\\
&=&\Phi'_1(M,A_e^i,A_d^i,Z^{i-1},Y^{i-1},X^i)\Phi_2(Y^i,S_0^i,X^i),
\eea
which implies markov chain
$(M,A_e^i,A_d^i,Z^{i-1})-(Y^{i-1},X^i,S_0)-(Y_i,S^i)$
which implies MC1.
\begin{lemma}
\label{cor1}
If MC1 holds,
\bea
P(y^N\parallel x^N,a_e^N,a_d^N,s_0)=P(y^N\parallel x^N,s_0).
\eea
\end{lemma}
\begin{proof}
This follows by chain rule in expanding $P(y^N\parallel x^N,a_e^N,a_d^N,s_0)$,
\bea
P(y^N\parallel x^N,a_e^N,a_d^N,s_0)&=&\prod_{i=1}^{N}
P(y_i|y^{i-1},x^i,a_e^i,a_d^i,s_0)\\
&\stackrel{(\ast)}{=}&\prod_{i=1}^{N}
P(y_i|y^{i-1},x^i,s_0)\\
&=&P(y^N\parallel x^N,s_0),
\eea
where ($\ast$) follows from MC1.
\end{proof}
To prove MC2, again summing over
$(\hat{M},X_{i+1}^{N},Z_i^{N},S_{i}^N,Y_{i}^N,A_{e,i+1}^N,A_{d,i}^N)$ in
Eq. (\ref{mainjoint1}) we obtain,
\bea
&&P_{M,A_{e}^i,A_d^{i-1},Z^{i-1},X^i,S_0^{i-1},Y^{i-1}}(m,a_{e}^i,a_d^{i-1},z^{
i-1
} , x^i ,
s_0^{i-1},y^{i-1})\\
&=&\frac{1}{\card{\mathcal{M}}}\1_{\{x_i=f_{e,i}(m,z^{i-1})\}}\1_{\{a_{e,i}=f_{
A_{e,i}}(m,z^{i-1})\}} \prod_{j=1}^{i-1}
\1_{\{a_{e,j}=f_{A_{e,j}}(m,z^{j-1})\}}\1_{\{x_j=f_{e,j}(m,z^{j-1})\}}
\nonumber\\
&&\times\prod_{j=1}^{i-1}\1_{\{a_{d,j}=f_{A_{d,j}}(y^{j-1)}\}}P(y_j,s_j|x_j,s_{
j-1 })
\1_{\{z_{i
}=f(a_{e,i},a_{d,i},y_i)\}}\times P_S(s_0)\\
&=&\Phi_1(M,X_i,A_{e,i},X^{i-1},A_e^{i-1},Z^{i-1})\Phi_2(X^{i-1},A_e^{i-1},Z^{
i-1},A_d^{ i-1 },S_0^{i-1},Y^{i-1}),
\eea
which implies the markov chain,
$(M,X_i,A_{e,i})-(X^{i-1},A_e^{i-1},Z^{i-1})-(Y^{i-1},S_0^{i-1},A_d^{i-1})$
which
implies MC2.
\par
To prove MC3, we sum over
$(\hat{M},X_{i+1}^{N},Z_i^{N},S_{i}^N,Y_{i}^N,A_{e,i+1}^N,A_{d,i+1}^N)$ in
Eq. (\ref{mainjoint1}) and obtain,
\bea
&&P_{M,A_{e}^i,A_d^{i},Z^{i-1},X^i,S_0^{i-1},Y^{i-1}}(m,a_{e}^i,a_d^{i},z^{i-1
} , x^i ,
s_0^{i-1},y^{i-1})\nonumber\\
&=&\frac{P_S(s_0)}{\card{\mathcal{M}}}\1_{\{x_i=f_{e,i}(m,z^{i-1}),a_{e
,i}=f_{
A_{e,i}}(m,z^{i-1})\}} \prod_{j=1}^{i-1}
\1_{\{a_{e,j}=f_{A_{e,j}}(m,z^{j-1}),x_j=f_{e,j}(m,z^{j-1}),z_{i
}=f(a_{e,i},a_{d,i},y_i)\}}P(y_j,
s_j|x_j,s_{
j-1 })
\nonumber\\
&&\times\prod_{j=1}^{i-1}\1_{\{a_{d,j}=f_{A_{d,j}}(y^{j-1)}\}}\\
&=&\Phi_1(M,S_0^{i-1},X^i,A_e^i,Z^{i-1},A_d^{i-1},Y^{i-1})\Phi_2(A_{d}^{i-1},Y^{
i-1 } , A_ { d , i } ),
\eea
which implies markov chain,
$A_{d,i}-(Y^{i-1},A_d^{i-1})-(X^i,A_e^i,Z^{i-1},S_0^{i-1},M)$, which implies
MC3.
\begin{lemma}
 \label{cor2}
If MC2 and MC3 holds, then,
\bea
Q(x^N,a_e^N,a_d^N\parallel y^{N-1},s_0)=Q(x^N,a_e^N\parallel
z^{N-1})Q(a_d^N\parallel y^{N-1}).
\eea
\end{lemma}
\begin{proof}
 Applying chain rule for causal conditioning,
\bea
Q(x^N,a_e^N,a_d^N\parallel y^{N-1},s_0)&=&\prod_{i=1}^N
Q(x_i,a_{e,i},a_{d,i}|x^{i-1},a_e^{i-1},a_d^{i-1},y^{i-1},s_0)\\
&=&Q(x_i,a_{e,i}|x^{i-1},a_e^{i-1},a_d^{i-1},y^{i-1},s_0)Q(a_{d,i}|x^{
i},a_e^{i},a_d^{i-1},y^{i-1},s_0)\\
&\stackrel{(\ast)}{=}&Q(x_i,a_{e,i}|x^{i-1},a_e^{i-1},a_d^{i-1},y^{i-1},z^{i-1} 
,s_0)Q(a_{d,i}|x^{
i},a_e^{i},a_d^{i-1},y^{i-1},s_0)\\
&\stackrel{(\ast\ast)}{=}&Q(x_i,a_{e,i}|x^{i-1},a_e^{i-1},z^{i-1})Q(a_{d,i}|a_d^
{
i-1},y^{i-1})\\
&=&Q(x^N,a_e^N\parallel
z^{N-1})Q(a_d^N\parallel y^{N-1})P(y^N\parallel x^N),
\eea
where ($\ast$) follows from the fact that, $z_i=f(a_{e,i},a_{d,i},y_i)$, while
($\ast\ast$) follows from MC2 and MC3.
\end{proof}
\begin{lemma}
 \label{cor3}
If a given scheme satisfies joint as in Eq. (\ref{mainjoint}) then,
\bea
P(s_0,x^N,a_e^N,a_d^N,y^N)&=&P(s_0)Q(x^N,a_e^N\parallel
z^{N-1})Q(a_d^N\parallel y^{N-1})P(y^N\parallel x^N,s_0)\\
P(x^N,a_e^N,a_d^N,y^N)&=&Q(x^N,a_e^N\parallel
z^{N-1})Q(a_d^N\parallel y^{N-1})P(y^N\parallel x^N).\\
\eea
\end{lemma}
\begin{proof}
Using Property 1 in Appendix \ref{basic}, we have, 
\bea
P(s_0,x^N,a_e^N,a_d^N,y^N)=P(s_0)Q(x^N,a_e^N,a_d^N\parallel
y^{N-1},s_0)P(y^N\parallel
x^N,a_e^N,a_d^N,s_0),
\eea
but already proved that MC1, MC2 and MC3 holds which implies by Lemmas
\ref{cor1} and \ref{cor2} that,
\bea
P(y^N\parallel x^N,a_e^N,a_d^N,s_0)&=&P(y^N\parallel x^N,s_0)\\
Q(x^N,a_e^N,a_d^N\parallel y^{N-1},s_0)&=&Q(x^N,a_e^N\parallel
z^{N-1})Q(a_d^N\parallel y^{N-1})P(y^N\parallel x^N),
\eea
which implies,
\bea
P(s_0,x^N,a_e^N,a_d^N,y^N)&=&P(s_0)Q(x^N,a_e^N\parallel
z^{N-1})Q(a_d^N\parallel y^{N-1})P(y^N\parallel x^N,s_0).
\eea
Summing over $s_0$,
\bea
P(x^N,a_e^N,a_d^N,y^N)&=&\sum_{s_0}P(s_0,x^N,a_e^N,a_d^N,y^N)\\
&=&Q(x^N,
a_e^N\parallel
z^{N-1})Q(a_d^N\parallel y^{N-1})\sum_{s_0}P(s_0)P(y^N\parallel x^N,s_0),
\eea
where last equality follows from Eq. (\ref{channelmodel1}).
\end{proof}
\section{Derivation of BAA-Action}
\label{BAA-Algorithm}
Here we will derive the BAA-Action algorithm (Algorithm 1) and also state the convergence results. We will state   similar Lemmas as in \cite{Naiss_Permuter_BAA}. Some proofs
are completely omitted as they follow verbatim from the corresponding lemmas in \cite{Naiss_Permuter_BAA}.
\begin{lemma} 
\label{lemma_baa_1}
For a fixed $\lambda\ge 0$, $\mathcal{I}(\mathbf{r},\mathbf{q},\lambda)$ is
concave, continuous, and has continuous partial derivatives in $\mathbf{r}$ and
$\mathbf{q}$.
\end{lemma}
\begin{proof}
$\mathcal{I}(\mathbf{r},\mathbf{q})$ has continuous partial derivatives in
$\mathbf{r}$ and $\mathbf{q}$ (follow the proof of Lemma 2 in
\cite{Naiss_Permuter_BAA} with $r(x^N\parallel y^{N-1})$ replaced by
$r(x^N,a^N\parallel z^{N-1})$ and $q(x^N|y^N)$ replaced by  $q(x^N,a^N|y^N)$.
$\E[\Lambda(A^N)]$ depends only on $\mathbf{r}$ and is linear in $\mathbf{r}$.
\end{proof}
\begin{lemma}
\label{lemma_baa_2}
For a fixed $\mathbf{r}$, $\mathbf{q}^*=\argmax_{\mathbf{q}}\
\mathcal{I}(\mathbf{r},\mathbf{q},\lambda)$, where
\bea
\mathbf{q}^*=q^*(x^N,a^N|y^N)=\frac{r(x^N,a^N\parallel z^{N-1})p(y^N\parallel
x^N)}{\sum_{x^N,a^N}r(x^N,a^N\parallel z^{N-1})p(y^N\parallel x^N)}.
\eea
\end{lemma}
\begin{proof}
Note that for a fixed $\mathbf{r}$, $\lambda\E[\Lambda(A^N)]$ is fixed.
It is easy to prove, as in Lemma 4 in \cite{Naiss_Permuter_BAA}, that
$\mathcal{I}(\mathbf{r},\mathbf{q}^*)-\mathcal{I}(\mathbf{r},\mathbf{q})\ge 0$
for all $\mathbf{q}$.
\end{proof}
From Lemma \ref{lemma_baa_2}, the following corollary is immediate, 
\begin{corollary}
\label{lemma_baa_3}
\bea
\max_{\mathbf{r}}\mathcal{I}(\mathbf{r},\mathbf{q},\lambda)&=&\max_{\mathbf{r}}
\max_{\mathbf{q}}\mathcal{I}(\mathbf{r},\mathbf{q},\lambda)\eea
\end{corollary}

Let the set $\Ascr_{i,z}=\{y^{i-1}:f(a^{i-1},y^{i-1})=z^{i-1}\}$, $\forall\
i\in[1:N]$. After proving the update of $q$ as in Lemma \ref{lemma_baa_2} above, the main theorem in derivation of algorithm is the update rule for $r$ which is as follows.
\begin{theorem}
\label{lemma_baa_4}
Fix $\lambda\ge 0$, and $\mathbf{q}=q(x^N,a^N|y^N)$, then
$\mathbf{r}^*=\argmax_{\mathbf{r}}\mathcal{I}(\mathbf{r},\mathbf{q},\lambda)$,
where
\bea
\mathbf{r^*}&=&r^*(x^N,a^N\parallel z^{N-1})\\
&=&\prod_{i=1}^{N}r(x_i,a_i|x^{i-1},a^{i-1},z^{i-1})\\
r(x_i,a_i|x^{i-1},a^{i-1},z^{i-1})&=&\frac{r'(x^i,a^i,z^{i-1})}{\sum_{x_i,a_i}
r'(x^i,a^i,z^{i-1})},
\eea
where
\bea
r'(x^i,a^i,z^{i-1})=\prod_{x_{i+1}^N,a_{i+1}^N,y_i^N}\prod_{\Ascr_{i,z}}\Bigg{[}
\frac{q(x^N,a^N|y^N)2^{-N\lambda\Lambda(a^N)}}{\prod_{j=i+1}^Nr(x_j,
a_j|x^{j-1},a^{j-1},z^{j-1})}\Bigg{]}^{\frac{p(y^N\parallel
x^N)\prod_{j=i+1}^{N}r(x_j,a_j\parallel
x^{j-1},a^{j-1},z^{j-1})}{\sum_{A_{i,z}}\prod_{j=1}^{i-1}p(y_j|x^j,y^{j-1})}}
\eea
\end{theorem}
\begin{proof}
Here we fix $\mathbf{q}$ and try to find $\mathbf{r}=\prod_{i=1}^N r_i$,
(denote $r_i\stackrel{\triangle}{=}r(x_i,a_i|x^{i-1},a^{i-1},z^{i-1})$, and
$p_i\stackrel{\triangle}{=}p(y_i|x^i,y^{i-1})$, $\forall\ i=[1:N]$), that
maximizes the expression
$\mathcal{I}(\mathbf{r},\mathbf{q})-\lambda\E[\Lambda(A^N)]$. Note as
there is one to one correspondence between $\mathbf{r}$ and the factors
$\{r_i\}_{i=1}^N$, maximizing $\mathcal{I}(\mathbf{r},\mathbf{q},\lambda)$ over
$\mathbf{r}$ is equivalent to maximizing over the factors $\{r_i\}_{i=1}^N$.
Note that already proved, $\mathcal{I}(\mathbf{r},\mathbf{q},\lambda)$ is
concave in $\mathbf{r}$, thus is concave in $r_i$ if all the other factors are
kept constant. Since constraints are linear, i.e., $\sum_{x_i,a_i}r_i=1$, by
concavity it follows we can use Lagrangian Multiplier method and Karush Kuhn
Tucker conditions to perform this maximization from $i=N$ to $i=1$ maximizing
the Lagrangian (described shortly) in a single factor $r_i$ one at a time. The
proof has similar steps as in Appendix A in \cite{Naiss_Permuter_BAA}, except
there is a different Lagrangian, 
\bea
J&=&\frac{1}{N}\sum_{x^N,a^N,y^N}\Bigg{(}p(y^n\parallel
x^N)\prod_{j=1}^{N}r_j\bigg{[}\log\frac{q(x^N,a^N|y^N)}{\prod_{j=1}^{N}r_j}
-\lambda N\Lambda(a^N)\bigg{]}\Bigg{)}\nonumber\\
&+&\sum_{i=1}^N\Bigg{(}\sum_{x^{i-1},a^{i-1},z^{i-1}}\nu_{i,x^{i-1},a^{i-1},z^{i
-1} } \bigg { ( } \sum_ { x_i , a_i } r_i-1\bigg{)}\Bigg{)}\\
&=&\frac{1}{N}\sum_{x^N,a^N,y^N}\Bigg{(}p(y^n\parallel
x^N)\prod_{j=1}^{N}r_j\bigg{[}\log\frac{q(x^N,a^N|y^N)2^{
-N\lambda\Lambda(a^N)}}{\prod_{j=1}^{N}r_j}\bigg{]}\Bigg{)}\nonumber\\
&+&\sum_{i=1}^N\Bigg{(}\sum_{x^{i-1},a^{i-1},z^{i-1}}\nu_{i,x^{i-1},a^{i-1},z^{i
-1} }\bigg{(}\sum_{x_i,a_i}
r_i-1\bigg{)}\Bigg{)},
\eea
where $\nu_{i,x^{i-1},a^{i-1},z^{i
-1} }>0$ are lagrangian multipliers. Hence for every,
$i\in\{1,\cdots,N\}$ we have,
\bea
\frac{\partial J}{\partial r_i}
&=&\frac{1}{N}\sum_{x_{i+1}^N,a_{i+1}^N,y_i^N}\sum_{A_{i,z}}\Bigg{(}p(y^n\parallel
x^N)\prod_{j=1,j\neq i}^{N}r_j\bigg{[}\log\frac{q(x^N,a^N|y^N)2^{
-N\lambda\Lambda(a^N)}}{\prod_{j=1}^{N}r_j}-1\bigg{]}\Bigg{)}+\nu_{i,x^{i-1},
a^{i-1},z^{i
-1} }\\
&=&\frac{1}{N}\sum_{x_{i+1}^N,a_{i+1}^N,y_i^N}\sum_{A_{i,z}}\Bigg{(}p(y^n\parallel
x^N)\prod_{j=1,j\neq i}^{N}r_j\bigg{[}\log\frac{q(x^N,a^N|y^N)2^{
-N\lambda\Lambda(a^N)}}{\prod_{j=1}^{i+1}r_j}-\log(\prod_{j=1}^{i-1}r_i)-\log
(r_i)-1\bigg{]}\Bigg{)}\nonumber\\ 
&&+\ \nu_ { i , x^ { i-1 } ,
a^{i-1},z^{i
-1} }\\
&=&0.
\eea
As for a given $r_i$, $(x^{i-1},a^{i-1},z^{i-1})$ are constants, and the term
$\prod_{j=1}^{i-1}r_i$ is constant and indepedent of $A_{i,z}$, and hence can be
taken out of the summation and can divide the whole equation to get a new
$\nu^{\ast}(i,x^{i-1},a^{i-1},z^{i-1})$ since
$\nu^{}(i,\cdot)$ is a function of $(x^{i-1},a^{i-1},z^{i-1})$. Also the other
three terms $(\log (\prod_{j=1}^{i-1}r_j), \log r_i, 1)$ are constants with
their coefficient being, 
\bea
\sum_{x_{i+1}^N,a_{i+1}^N,y_i^N}\sum_{A_{i,z}}\Bigg{(}p(y^n\parallel
x^N)\prod_{j=i+1}^{N}r_j\Bigg{)}=\sum_{A_{i,z}}\prod_{j=1}^{i-1}p_j.
\eea
Rest of the proof in Appendix A in
referenced paper follows verbatim with $q(\cdot)$ inside the logarithm being
replaced with $q(\cdot)2^{
-N\lambda\Lambda(a^N)}$.
\end{proof}
Thus following the above lemmas we have similar Blahut-Arrimoto-Algorithm of
alternating maximization. The above lemmas similarly as in
\cite{Naiss_Permuter_BAA} provide natural lower bounds $I_L^{(k)}(\lambda)$
indexed by $k$, the number of iterations and $\lambda$, the lagrangian
multiplier, given by,
\bea
I_L^{(k)}(\lambda)=\frac{1}{N}\sum_{x^N,a^N,y^N}r^{(k)}(x^N,a^N\parallel
z^{N-1})p(y^n\parallel x^N)\log
\frac{q^{(k)}(x^N,a^N|y^N)}{r^{(k)}(x^N,a^N\parallel z^{N-1})}.
\eea
Due to Lemma \ref{lemma_baa_1} above, which states, $\mathcal{I}(\mathbf{r},\mathbf{q},\lambda)$ is
concave, continuous, and has continuous partial derivatives in $\mathbf{r}$ and
$\mathbf{q}$, the alternating maximization procedure converges (cf. proof of Lemma 1 in \cite{Naiss_Permuter_BAA}), 
\bea
I_L^{(k)}(\lambda)\uparrow\max_{\mathbf{r}}\Bigg{\{}\frac{1}{N}I(X^N\rightarrow
Y^N)-\lambda\E[\Lambda(A^N)]\Bigg{\}}\stackrel{\triangle}{=}
C_N^{(\lambda)},
\eea
where $\uparrow$ implies convergence from below as $k\rightarrow\infty$. 
Similarly there exist upper bounds,
\bea
I_U^{(k)}(\lambda)=\frac{1}{N}\max_{x_1,a_1}\sum_{y_1}\max_{x_2,a_2}\cdots\max_{
x_N , a_N } \sum_ { y_N}p(y^N\parallel x^N)\log\frac{p(y^N\parallel
x^N)2^{
-N\lambda\Lambda(a^N)}}{\sum_{x^N,a^N}p(y^N\parallel
x^N)r^{(k)}(x^N,a^N\parallel z^{N-1})}.
\eea
To describe the mentioned nature of upper bound $I_U^{(k)}(\lambda)$, we have
the following lemmas corresponding to those in \cite{Naiss_Permuter_BAA} (most
proofs are verbatim and hence omitted or described briefly).
\begin{lemma}
\label{lemma_baa_5}
Let $I_{r_1}(X^N\rightarrow Y^N)$ correspond to $r_1(x^N,a^N\parallel z^{N-1})$
and $\E_{r_1}[\Lambda(A^N)]$ the corresponding cost incurred, then for any
$r_0(x^N,a^N\parallel z^{N-1})$ and $\lambda\ge 0$,
\bea
&&\frac{1}{N}I_{r_1}(X^N\rightarrow Y^N)-\lambda \E_{r_1}[\Lambda(A^N)]\nonumber\\
&&\le \sum_{x^N,a^N,y^{N-1}}r_1(x^N,a^N\parallel
z^{N-1})\sum_{y_N}p(y^N\parallel x^N)\log \frac{p(y^N\parallel
x^N)2^{
-N\lambda\Lambda(a^N)}}{ \sum_{x^N,a^N}r_0(x^N,a^N\parallel
z^{N-1})p(y^N\parallel x^N)}\nonumber\\
\eea
\end{lemma}
\begin{proof}
Consider, 
\bea
&&\frac{1}{N}\sum_{x^N,a^N,y^{N-1}}r_1(x^N,a^N\parallel z^{N-1})\sum_{y_N}p(y^N\parallel
x^N)\log \frac{p(y^N\parallel x^N)2^{
-N\lambda\Lambda(a^N)}}{
\sum_{x^N,a^N}r_0(x^N,a^N\parallel z^{N-1})p(y^N\parallel x^N)}\nonumber\\
&&-\frac{1}{N}I_{r_1}(X^N\rightarrow Y^N)+\lambda \E_{r_1}[\Lambda(A^N)]\\
&=&\frac{1}{N}\sum_{x^N,a^N,y^{N}}r_1(x^N,a^N\parallel z^{N-1})p(y^N\parallel x^N)\log
\frac{ \sum_{x^N,a^N}r_1(x^N,a^N\parallel z^{N-1})p(y^N\parallel x^N)}{
\sum_{x^N,a^N}r_0(x^N,a^N\parallel z^{N-1})p(y^N\parallel x^N)}\\
&=&D(p_1(y^N)\parallel p_0(y^N))\ge 0.
\eea
\end{proof}
\begin{lemma}
\label{lemma_baa_6}
For every $\lambda\ge 0$,
\bea
C_N^{(\lambda)}\le
\frac{1}{N}\min_{r_0}\max_{x_1,a_1}\sum_{y_1}\max_{x_2,a_2}\cdots\max_{
x_N , a_N } \sum_ { y_N}p(y^N\parallel x^N)\log\frac{p(y^N\parallel
x^N)2^{
-N\lambda\Lambda(a^N)}}{\sum_{x^N,a^N}p(y^N\parallel
x^N)r_0(x^N,a^N\parallel z^{N-1})}
\eea
\end{lemma}
\begin{proof}
The proof follows verbatim from the proof of Lemma 9 in
\cite{Naiss_Permuter_BAA}.
\end{proof}
\begin{lemma}
\label{lemma_baa_7}
The upper bound in Lemma \ref{lemma_baa_6} is tight and is obtained by
$r_0(x^N,a^N\parallel z^{N-1})$ that achieves the capacity.
\end{lemma}
\begin{proof}
The proof follows in line to the proof of Lemma 10 in \cite{Naiss_Permuter_BAA}
except that there is an additional term $-\lambda\E[\Lambda(A^N)]$ in
the Lagrangian which accounts for the term $2^{
-N\lambda\Lambda(a^N)}$
in the right hand side of expression in Lemma \ref{lemma_baa_6}.
\end{proof}
Thus we have, $I_U^{(k)}(\lambda)\downarrow C_N^{(\lambda)}$ as
$k\rightarrow\infty$, $\forall\ \lambda\ge 0$, where the down arrow $\downarrow$
implies convergence from above. Since the Lagrangian multiplier $\lambda$ characterizes the tradeoff, the corresponding point on the tradeoff curve is $(C_N^{(\lambda)},\Gamma^{(\lambda)})$, where 
\bea
\Gamma^{(\lambda)}= \sum_{x^N,a^N,y^N}\mathbf{r}^*(x^N,a^N\parallel
z^{N-1})p(y^N\parallel x^N)\Lambda(A^N).
\eea
\section{Proof of Theorem \ref{theorem-evaluation}}
\label{evaluation}
For convenience, instead of $N$, consider the block length to be $B$ divided
into $M$ sub-blocks each of length $N$.  Define, 
\bea
C_B(\Gamma)&\stackrel{\triangle}{=}&\max \frac{1}{B} I(X^B\rightarrow Y^B,S^B|S_0=0)
\eea
where $\max$ is on appropriate joint distribution that satisfies the cost constraint. Now we have, 
\bea
C_B(\Gamma)&=& \max \frac{1}{B} \sum_{i=1}^B I(X^i;Y_i,S_i|Y^{i-1},S^{i-1},S_0=0)\\
&=& \max \frac{1}{B} \sum_{i=1}^B H(Y_i,S_i|Y^{i-1},S^{i-1},S_0=0)- H(Y_i,S_i|X^i,Y^{i-1},S_0=0)\\
&\stackrel{(a)}{=}& \max \frac{1}{B} \sum_{i=1}^B H(Y_i,S_i|Y^{i-1},S^{i-1},S_0=0)-
H(Y_i,S_i|X_i,S_{i-1})\\
&=& \max \frac{1}{B} \sum_{j=1}^M\sum_{i=N(j-1)+1}^{Nj} H(Y_i,S_i|Y^{i-1},S^{i-1},S_0=0)-
H(Y_i,S_i|X_i,S_{i-1})\\
&\le& \max \frac{1}{B} \sum_{j=1}^M\sum_{i=N(j-1)+1}^{Nj}
H(Y_i,S_i|Y^{i-1}_{N(j-1)+1},S^{i-1}_{N(j-1)})- H(Y_i,S_i|X_i,S_{i-1})\\
&\stackrel{}{\le}&\frac{1}{B}\max  \sum_{j=1}^M
I(X_{N(j-1)+1}^{Nj}\rightarrow Y_{N(j-1)+1}^{Nj},S_{N(j-1)+1}^{Nj}|S_{N(j-1)})\\
&\stackrel{}{\le}&\frac{1}{B} \sum_{j=1}^M \max 
I(X_{N(j-1)+1}^{Nj}\rightarrow Y_{N(j-1)+1}^{Nj},S_{N(j-1)+1}^{Nj}|S_{N(j-1)})\\
&\stackrel{}{\le}&\frac{1}{B} \sum_{j=1}^M \max 
\Bigg{(}\sum_{s\in\mathcal{S}}p(S_{N(j-1)}=s)I(X_{N(j-1)+1}^{Nj}\rightarrow
Y_{N(j-1)+1}^{Nj},S_{N(j-1)+1}^{Nj}|S_{N(j-1)}=s)\Bigg{)}\\
&\stackrel{(b)}{\le}&\frac{1}{B} \sum_{j=1}^M  
\Bigg{(}\sum_{s\in\mathcal{S}}p(S_{N(j-1)}=s)\max I(X_{N(j-1)+1}^{Nj}\rightarrow
Y_{N(j-1)+1}^{Nj},S_{N(j-1)+1}^{Nj}|S_{N(j-1)}=s)\Bigg{)}\\
&\stackrel{(c)}{=}& \frac{1}{B} \sum_{j=1}^M\max I(X_{N(j-1)+1}^{Nj}\rightarrow
Y_{N(j-1)+1}^{Nj},S_{N(j-1)+1}^{Nj}|S_{N(j-1)}=0)
\eea
where (a) follows from the channel definition, (b) follow from that fact
that $\max (w_1f(x)+w_2g(x))\le w_1\max f(x) +w_2 \max g(x) $, and (c) follows
from the symmetry of the channel structure, $\mathcal{S}=\{0,1\}$ here. But note that $\frac{1}{N}\max  I(X_{N(j-1)+1}^{Nj}\rightarrow
Y_{N(j-1)+1}^{Nj},S_{N(j-1)+1}^{Nj}|S_{N(j-1)}=0)= C_N(\Gamma_j)$, where $\Gamma_j$ are cost
incurred in each block such that, $\frac{N}{B}\sum_{j=1}^{M}\Gamma_j\le \Gamma$.
Thus we have, 
\bea
C_B(\Gamma)&\stackrel{}{\le}& \frac{N}{B} \sum_{j=1}^M C_{N}(\Gamma_j)\\
&\stackrel{(g)}{\le}&C_{N}(\frac{N}{B}\sum_{j=1}^{M}\Gamma_j)\\
&\stackrel{(h)}{\le}&C_{N}(\Gamma),
\eea
where (g) and (h) finally follow from the concavity and non-increasing nature of
$C_N(\Gamma)$. Note the above holds for any
$N$ thus with $B\rightarrow\infty$, we obtain,
$C(\Gamma)\le C_N(\Gamma)$ $\forall\ N$.
\par
We will now derive the lower bound. Here also assume the block length is $B$, which is divided into $M$ sub-blocks
of length $N$. The following achievability scheme is used sub-block by
sub-block. In the last time epoch of each sub-block of length $N$, action is
taken to observe the feedback. This feedback provides the initial state for rest
of the $N$ channel uses in the next sub block where encoder encodes to achieve,
$C_N(\Gamma)$. Thus the total incurred cost is at most
$\frac{M+MN\Gamma}{MN}=\Gamma+\frac{1}{N}$. Thus we have $C_N(\Gamma)\le
C(\Gamma+\frac{1}{N})$ or $C_N(\Gamma-\frac{1}{N})\le C(\Gamma)$ for $\Gamma \in[
\frac{1}{N},1]$. 
\end{document}